\documentclass[11pt]{article}

\usepackage[margin=1in]{geometry}
\usepackage[comma]{natbib}
\usepackage{anysize}
\usepackage{color}
\usepackage{graphicx}
\usepackage{amsmath,amssymb,amsthm,amsopn,mathrsfs}
\usepackage[colorlinks=true, linkcolor=red, urlcolor=blue, citecolor=blue]{hyperref}
\usepackage[linesnumbered,ruled,vlined]{algorithm2e}
\usepackage{rotating} % to make sideway tables/figures
\usepackage{pdflscape}
\usepackage{afterpage}

\newcommand{\der}{\mathrm{d}}
\newcommand{\indi}{\mathbb{I}}
\def\cS{\mathcal{S}}
\def\der{\text{d}}
\def\real{{\mathbb R}}
\def\expect{{\mathbb E}}
\def\bba{{\bf a}}
\def\bs{{\bf s}}
\def\bt{{\bf t}}
\def\bu{{\bf u}}
\def\bv{{\bf v}}
\def\bz{{\bf z}}
\def\diag{\mbox{diag}}

\DontPrintSemicolon
\LinesNotNumbered
\newcommand{\To}{\mbox{\upshape\bfseries to}}
\newcommand{\nextnr}{\stepcounter{AlgoLine}\ShowLn}
\SetKwRepeat{Do}{do}{\nextnr while}

\theoremstyle{definition}\newtheorem{defi}{Definition}
\theoremstyle{plain}\newtheorem{prop}{Proposition}\theoremstyle{plain}\newtheorem{lem}{Lemma}
\theoremstyle{remark}

\begin{document}

\title{High-dimensional inference using the extremal skew-$t$ process}

\author{B. Beranger\footnote{UNSW Data Science Hub \& School of Mathematics and Statistics, UNSW Sydney, Australia.}\;\,\footnote{Communicating Author: {\tt B.Beranger@unsw.edu.au}},\; A. G. Stephenson\footnote{Data61, CSIRO, Clayton, Australia.}, \; and\; S. A. Sisson$^*$
}

\date{}

\maketitle
\begin{abstract}
Max-stable processes are a popular tool for the study of environmental extremes, and the
extremal skew-$t$ process is a general model that allows for a flexible extremal dependence
structure. For inference on max-stable processes with high-dimensional data, exact
likelihood-based estimation is computationally intractable. Composite likelihoods, using lower dimensional components,
and Stephenson-Tawn likelihoods, using occurrence times of maxima,
are both attractive methods to circumvent this issue for moderate dimensions. 

In this article we establish the theoretical formulae for simulations of and inference for the extremal skew-$t$ process. We also incorporate the Stephenson-Tawn concept into the composite likelihood framework, giving greater statistical and computational efficiency for higher-order composite likelihoods. We compare 2-way (pairwise), 3-way (triplewise), 4-way, 5-way and 10-way composite likelihoods for models of up to 100 dimensions.
Furthermore, we propose cdf approximations for the Stephenson-Tawn likelihood function, leading to large computational gains, and enabling accurate fitting of models in large dimensions in only a few minutes. 
We illustrate our methodology with an application to a 90-dimensional temperature dataset from Melbourne, Australia.

\vspace*{2mm}
\end{abstract}
%

%%%%%%%%%%%%%%%%%%%%%%%%%%%%%%%%%%%%%%%%%%%%%
%
% SECTION 1
%
%%%%%%%%%%%%%%%%%%%%%%%%%%%%%%%%%%%%%%%%%%%%%
%
\section{Introduction}
\label{sec:intro}

In the current environmental context, modelling the extremes of natural processes is receiving ever growing attention \citep[see, e.g.,][]{davison2012b, cooley2012}. A sound knowledge of the extremal behaviour of temperature, precipitation or winds is crucial as these events often lead to catastrophes with a strong impact on human life. Such events are spatial by nature and max-stable processes are a convenient tool to analyse spatial extremes which can extrapolate beyond the observed data. 

Max-stable processes arise as the pointwise maxima of an infinite number of suitably normalised stochastic processes. 
Consider $Y_1, \ldots, Y_n$, $n$ independent replications of a real-valued stochastic process $\{ Y(s) \}_{s \in \cS}$ with continuous sample paths  on the spatial domain $\cS$, a compact subset of  $\real^k, k \geq1$, representing a $k$-dimensional region of interest. If there exists sequences of continuous functions $a_n(s) >0$ and $b_n(s) \in \real$ such that the rescaled pointwise maxima
$$
\max_{j=1,\ldots,n} \frac{ Y_j(s) - b_n(s)}{a_n(s)} 
$$
converge weakly as $n \rightarrow \infty$ to a process $Z(s)$, $s \in \cS$, 
with non-degenerate margins, then the limiting process $\{ Z(s)\}_{s \in \cS}$ is called a max-stable process
\citep[see][Ch.~9]{dehaan1984,dehaan2006}.

The construction of max-stable models is enabled by the spectral representation of max-stable processes of \citet{schlather2002a}, which extends the work of \citet{dehaan1984} to random functions and is defined as follows. 
Let $\{W(s)\}_{s \in \cS}$ be a real-valued stochastic process with continuous sample paths on 
$\cS$ such that
$$
\expect \left\{ \sup_{s\in\cS} W(s) \right\} < \infty,
\quad m_+(s) = \expect[ \{W(s)\}_+^{\nu} ] \, \in (0, \infty), \:\forall s\in \cS,
$$
for some fixed $\nu>0$, where $\{ W( \cdot )\}_+^\nu = \max\{W (\cdot) ,0\}^\nu$.
Let $\{\zeta_j\}_{j\geq1}$ be the points of an inhomogeneous Poisson point process on 
$(0,\infty)$ with intensity $\nu \zeta^{-(\nu+1)}$, $\nu>0$, which are independent of $W_1,W_2,\ldots$, which are independent copies of $W$.
If we define
\begin{equation}
\label{eq:max_stable}
Z (s) = \max_{j=1,2,\ldots} \zeta_j Z_j^*(s),\quad s \in \cS,
\end{equation}
with
\begin{equation}
\label{eq:max_stable_zstar}
Z_j^*(s) = \{ W_j(s)\}_+  / \{m_+(s)\}^{1/\nu},
\end{equation}
then $Z$ is a max-stable process with common $\nu$-Fr\'{e}chet univariate margins 
\citep{opitz2013}.
The spectral representation of \citet{dehaan1984} can be retrieved by setting $Z_j^*(s) = f(s-X_j)$, where $X_j$ are the points of an homogenous Poisson point process on $\real^k$ with intensity measure $\Lambda (\der x)$ and $f(\cdot)$ is a unimodal continuous probability density function.

For a finite set of spatial locations $\{s_i\}_{i=1 ,\ldots,d} \in \cS$, the finite-dimensional 
distribution of $Z(s)$ is given by
\begin{equation}
\label{eq:dist_fun_msp}
G( \bz ) 
\equiv \Pr \{ Z_i \leq z_i, i=1,\ldots,d \}
= \exp \left\{ - V(\bz) \right\}, 
\quad \bz = (z_1, \ldots, z_d) >0,
\end{equation}
where $V$ is a function defined as
$$
V(\bz) = \expect 
\left[ 
\max_{i=1,\ldots, d}  \frac{ \{ W(s_i)\}_{+}^{\nu} }{ z_i^{\nu} m_+(s_i) } 
\right],
$$
which fully characterises the dependence structure between extremes. 
It is referred to as the exponent function. 
If the margins are unit Fr\'{e}chet distributed, i.e.~$\nu =1$ in the above representation \eqref{eq:max_stable}, 
then $Z$ is referred to as a simple max-stable process.
The most widely used max-stable models include the well-known Gaussian extreme-value process commonly referred to as the Smith model \citep{smith1990a}, the Schlather or extremal Gaussian process \citep{schlather2002a}, the geometric Gaussian process \citep{davison2012b}, the Brown-Resnick process \citep{brown1977,kabluchko2009} and the extremal-$t$ \citep{opitz2013,nikoloulopoulos2009}. 

Motivated by the need for flexible models,  \citet{beranger2017} proposed a wide family of max-stable processes -- the extremal skew-$t$ process -- allowing for skewness in the dependence structure. There $W(s)$ is taken to be a skew-normal random field on $s \in \cS$ with finite $d$-dimensional distribution $SN_d(\bar{\Omega}, \alpha, \tau)$, with $\bar{\Omega}$, $\alpha \in \real^d$, $\tau$ respectively representing the correlation matrix, slant and extension parameters. Assuming unit-Fr\'{e}chet margins, the $d$-dimensional exponent function is given by
\begin{equation}
\label{eq:V_ext_st}
V(\bz)=
\sum_{i=1}^d z_i^{-1} \Psi_{d-1}
\left(\left(
\sqrt{\frac{\nu+1}{1-\rho^2_{i,j}}}
\left(
\frac{z^\circ_j}{z^\circ_i} - \rho_{i,j}
\right),j\in I_i\right)^\top; \bar{\Omega}^\circ_i, \alpha^\circ_i, \tau^\circ_i, \kappa^\circ_i,\nu+1
\right),
\end{equation}
with $z_i\equiv z(s_i)$, $z^\circ_j = (z_j m_{j+})^{1/\nu}$, $m_{j+} \equiv m_+(s_j)$. Furthermore $\Psi_{d-1}$ is a $(d-1)$-dimensional non-central extended skew-$t$ distribution with  correlation matrix $\bar{\Omega}^\circ_i$, shape $\alpha^\circ_i \in \real^{d-1}$, extension $\tau^\circ_i \in \real$, non-centrality $\kappa^\circ_i \in \real$ and $\nu+1$ degrees of freedom, where
$I=\{1,\ldots,d\}$, $I_i=I\backslash\{i\}$, and $\rho_{i,j}$ is the $(i,j)$-th element of 
$\bar{\Omega}$.
See Appendix~\ref{app:def_ncextst} for a definition of the non-central extended skew-$t$ 
distribution and Appendix~\ref{app:ext_st_expo} for the expression of $m_+(s)$ and additional details about the parameters.
It is easy to see that setting $\alpha$ to the zero vector and $\tau=0$ recovers the extremal-$t$ process and further fixing $\nu =1$ reduces to the Schlather model.
The extremal skew-$t$ process has the appealing characteristic of being non-stationary, 
where $\nu$ controls the level of overall dependence: smaller values indicate high dependence and vice versa \citep[see][]{beranger2015, nikoloulopoulos2009}.

The ability to exactly simulate data from a max-stable process is important for assessing the performance of inference procedures, and for making predictions under the fitted model.
Simulations can be used to evaluate the probability that an environmental field (temperature, precipitation, etc.) exceeds some critical level across some region ($\cS$) despite only 
being observed at a finite number of locations \citep{buishand2008, blanchet2011}.
Conditional simulations can be of interest, e.g.~for prediction, depending on the existence of constraints.
Unconditional process simulations play an important role in generating conditional simulations \citep{dombry2016}.

As defined above, max-stable processes arise as the pointwise maxima over an infinite number of random functions (c.f.~eq.~\ref{eq:max_stable}) which at first glance might 
seem to require the use of finite approximations.
\citet{schlather2002a} first proposed an exact simulation procedure by showing that for some models only a 
finite number of points $\{ \zeta_j\}_{j \geq 1}$ and stochastic processes $\{ W_j(s)\}_{j \geq 1}$
will contribute to the componentwise maxima.
More recently \citet{dieker2015} and \citet{thibaud2015} respectively developed exact simulation
procedures for the Brown-Resnick and extremal-$t$ processes. 
\citet{dombry2016} extended the approach from \citet{dieker2015} and used it for the simulation 
of max-stable processes using either the spectral measure or through the simulation of extremal 
functions,
the latter being computationally more efficient. 
Other recent work is also given by \citet{oesting2018} and \citet{liu2016}.
Conditional simulation of max-stable processes was first studied by \citet{wang2011}, and shortly afterwards \citet{dombry2013a} and \citet{dombry2013b} defined a general framework. 

In recent years the $d$-dimensional distribution functions of most of the widely used max-stable models have been made available. See e.g.~\cite{genton2011}  for the Smith model, \citet{huser2013} for the Brown-Resnick and \eqref{eq:V_ext_st} for the extremal skew-$t$. However, due to the exponential form of the distribution function \eqref{eq:dist_fun_msp}, as the dimension increases, there is an explosion in the number of terms in the likelihood function 
\begin{equation*}
\label{eq:L}
\mathrm{L}_d (\bz; \theta) = 
\exp \{ -V(\bz; \theta) \}
\times \sum_{\Pi \in \mathcal{P}_d} \prod_{k=1}^{|\Pi|} - V_{\pi_k}(\bz; \theta), 
\end{equation*}
where $\mathcal{P}_d$ is the set of all possible partitions of $\{1,\ldots,d\}$, each partition $\Pi \in \mathcal{P}_d$ has elements $\pi_k$, for $k=1, \ldots, |\Pi|$, and $V_{\pi_k}(\cdot)$ represents the partial derivatives of $V(\cdot)$ w.r.t $\pi_k$.
The cardinality of $\mathcal{P}_d$, the set of all possible partitions $\Pi$ of  $\{1, \ldots, d\}$, corresponds to the $d$-th Bell number, making full likelihood inference computationally intractable for high-dimensional data.

As a result, composite likelihood (CL) methods using pairs \citep{padoan2010, davison2012c, davison2012b} and triplets \citep{genton2011,huser2013} but also higher orders \citep{castruccio2016}, have been investigated. Under some mild conditions, CL estimators have been shown to be consistent and asymptotically normal \citep{padoan2010}, and thus are an attractive substitute to full likelihood estimation. Additionally \cite{sang2014} and \cite{castruccio2016} have suggested the use of weighted composite likelihood with binary weights in order to truncate the likelihood and solely conserve the most informative tuples. However, despite being consistent, CL estimators can have a low efficiency compared to full likelihood estimators \citep{huser2016}. CL efficiency has mainly been studied for pairs and triples \citep[e.g.][]{huser2016}. Only \citet{castruccio2016} consider higher orders and compare them to the full likelihood, but this is limited to models of up to dimension $d=11$.  
Under the assumption that the spectral random vectors (equivalent to the $Z_j^*$ given in \eqref{eq:max_stable_zstar}) of the multivariate max-stable distributions have known conditional distribution, \citet{bienvenue2016} demonstrate that the partial derivative of the exponent function can be written as univariate integrals, allowing for high-dimensional extremes modelling.

An alternative method developed by \citet{stephenson2005} produced a likelihood simplification when the time occurrences of each block maxima are recorded. The censored Poisson likelihood approach introduced by \citet{wadsworth2014} extends the Stephenson-Tawn (ST) likelihood, which can be seen as a special case, and highlights large efficiency gains. A drawback of the ST likelihood is the possibility of introducing bias, in particular when the times of occurrence (hitting scenarios) are not drawn from their limiting distribution \citep{wadsworth2015}, which is more likely when the number of dimensions is high relative to the number of events recorded. 
\citet{wadsworth2015} derived a second-order bias correction for moderately high dimensions. 
\cite{thibaud2016} considered a Bayesian hierarchical model using the Stephenson-Tawn likelihood which enables calculation of the full likelihood of the Brown-Resnick process. 
Recently \citet{huser2019} proposed a stochastic expectation-maximisation algorithm which rewrites the full likelihood as the sum of ST likelihoods. They provide numerical results for the Brown-Resnick model in dimension $d=10$, considering $10$ independent replicates of the process. \citet{whitaker2019} also proposed to aggregate data in order to use pairwise CL in up to $d=100$ dimensions.

Coupled with the development of new technologies, the hope of better understanding extreme phenomena has resulted in more abundant data and a need for well performing estimation procedures.
The aim of this work is to introduce some tools allowing the use of flexible max-stable process models, such as the extremal skew-$t$, and to establish an inferential methodology that permits the use of these models in high dimensions. We propose combining the CL and ST approaches in order to perform high-dimensional inference. This further reduces the computational burden of computing the likelihood function, although calculating high-dimensional cdfs is still required. The ST approach is statistically efficient if the information on the hitting scenario is drawn from the limiting distribution. For the cdfs involved in the likelihood we propose the use of quasi-Monte Carlo approximation methods (see Section \ref{ssec:STeval}).
This accordingly admits the possibility of fitting max-stable models in dimensions up to $d=100$ for the extremal-$t$ and extremal skew-$t$ models, within a relatively short computational timeframe. We note that \citet{davison2012c} measured the effect of including the event time information into the pairwise composite likelihood context, but did not investigate the possibility of using the Stephenson-Tawn likelihood for high-dimensional inference.
  
The remainder of this paper is organised as follows: Section~\ref{sec:msp} develops the procedure to perform exact and conditional simulations from the extremal skew-$t$ process (as required for assessing parameter estimation performance, and subsequent predictive inference), and derives the partial derivative of its exponent function in any dimension necessary for inference. Section~\ref{sec:sim} proposes some quasi-Monte Carlo approximations of the cdfs embedded in the ST likelihood. The trade-off between statistical and computational efficiency for the ST likelihood as well as for a combination of the ST and CL likelihoods is investigated via simulation studies. Section~\ref{sec:dataeg} provides an illustrative application to daily maximum temperatures in Melbourne, Australia, highlighting the need for flexible models and inferential procedures. Section~\ref{sec:discuss} concludes with a discussion.

%%%%%%%%%%%%%%%%%%%%%%%%%%%%%%%%%%%%%%%%%%%%
%
% SECTION 2
%
%%%%%%%%%%%%%%%%%%%%%%%%%%%%%%%%%%%%%%%%%%%%
%
\section{The extremal skew-\texorpdfstring{$t$}{Lg} process}
\label{sec:msp}
%

%
%
% SUBSECTION SIMULATION OF MAX STABLE PROCESSES
%
%

\subsection{Exact simulations}

Exact simulation of data from a max-stable process is necessary to properly evaluate any inference process for this model, and exact conditional simulation is important for process predictions under the fitted model.
In order to perform exact simulations of a max-stable process, Algorithm 2  of \citet{dombry2016} requires the simulation of random functions $Y_j(s) = Z_j^*(s)/Z_j^*(s_0)$ with distribution $P_{s_0}$ for $s_0 \in \cS$, where $Z_j^*(s)$ is defined in equation \eqref{eq:max_stable_zstar}. The following Proposition establishes the distribution $P_{s_0}$ required to sample from the extremal skew-$t$ model.
\begin{prop}
\label{prop:p0_ext_st}
Consider the extremal skew-$t$ process defined in Section~\ref{sec:intro} with some covariance function $K$. For all $s_0 \in \cS$, the distribution $P_{s_0}$ is equal to the distribution of $T_+^\nu m_{0+} / m_+$, where $m_+ = (m_{1+}, \ldots, m_{d+})$, $T = \{ T(s) \}_{s \in \cS}$ is an extended skew-$t$ process with location and scale functions
$$
\mu(s) = K(s_0, s) 
\quad\mbox{and}\quad \hat{K}(s_1, s_2) = \frac{K(s_1, s_2) - K(s_0, s_1) K(s_0, s_2)}{\nu +1}, $$
slant vector $\alpha$, extension $(\alpha_0 + \alpha^\top \Sigma_{d;0}) \sqrt{\nu +1}$, non-centrality $-\tau$ and $\nu + 1$ degrees of freedom.
\end{prop}
\begin{proof}
See Appendix~\ref{app:exact_ext_st}
\end{proof}

\subsection{Conditional simulations}
The algorithm provided in \citet[Theorem~1]{dombry2013b} is a three-step procedure for conditional simulation of max-stable processes. This methodology relies on the knowledge of the conditional intensity function, defined as follows.
Assuming unit Fr\'echet margins, the spectral representation \eqref{eq:max_stable} can be 
rewritten as
\begin{equation*}
\label{eq:max_stable_unit}
Z(s) = \max_{j = 1,2 ,\ldots } \{ \zeta_j W_j(s)^\nu\} / \{ m_+(s) \} 
= \max_{j = 1,2 ,\ldots } \varphi_j,
\end{equation*}
where $\{ \zeta_j \}_{j \geq 1}$ are the points of an homogeneous Poisson point process on 
$(0, \infty)$ with intensity $\der \Lambda(\zeta) = \zeta^{-2} \der \zeta$.
Let $\bs = (s_1, \ldots, s_d) \in \cS^d$. For all Borel sets $A \subset \real^d$, the Poisson point 
process $\Phi = \{ \varphi_j\}_{j \geq 1}$ on $\mathcal{C} = \mathcal{C} \{\cS \}$ the space of 
continuous real-valued functions on $\cS$, 
has intensity measure
\begin{equation}
\label{eq:def_intensity_measure}
\Lambda_\bs(A) 
= \int_0^\infty \Pr \{ \zeta W(\bs)^\nu / m_+(\bs) \in A \} \zeta^{-2} \der \zeta
= \int_A \lambda_\bs (\bv) \der \bv.
\end{equation}
Note that in the above representation of the max-stable process $Z(s)$, $W_+(s)$ is replaced 
by $W(s)$ so that the point process $\Phi$ is regular, i.e.~$\Lambda_\bs(\der \bz) = \lambda_\bs (\bz) \der \bz$ for all $\bs \in \cS^d$.
The conditional intensity function is then given by
\begin{equation}
\label{eq:def_cond_int}
\lambda_{\bt | \bs, \bv } (\bu) = \frac{\lambda_{(\bt, \bs)}  (\bu, \bv) }{\lambda_\bs (\bv)},
 \quad (\bt, \bs) \in \cS^{m+d}, \quad (\bu, \bv) \in \real^m \times \real^d_+.
\end{equation}
\citet{dombry2013b} give the closed form expression of the conditional intensity function for the 
Brown-Resnick and Schlather models, whereas \citet{ribatet2013} derive those of the 
extremal-$t$.
The following Proposition provides the conditional intensity function for the extremal skew-$t$
model. 
\begin{prop}
\label{prop:cond_intensity_ext_st}
Consider the representation of the extremal skew-$t$ process in Section \ref{sec:intro} at 
$(\bt,\bs) \in \cS^{m+d}$ with slant $\alpha_{(\bt,\bs)} = (\alpha_\bt, \alpha_\bs) \in \real^{m+d}$ 
and extension parameter $\tau_{(\bt,\bs)} \in \real$.
Provided the correlation matrix 
$
\bar{\Omega}_{(\bt,\bs)} = \left[
\begin{array}{cc}
\bar{\Omega}_\bt & \bar{\Omega}_{\bt\bs} \\
\bar{\Omega}_{\bs\bt} & \bar{\Omega}_\bs
\end{array}
\right]
$
is positive definite,
the conditional intensity function \eqref{eq:def_cond_int} is given by
$$
\lambda_{\bt | \bs, \bv}(\bu) = 
\psi_m \left(\bu^\circ; \mu_{\bt | \bs, \bv}, \Omega_{\bt | \bs, \bv}, \alpha_{\bt | \bs, \bv}, 
\tau_{\bt | \bs, \bv}, \kappa_{\bt | \bs, \bv}, \nu_{\bt | \bs, \bv} \right)
\nu^{-m} \prod_{i=1}^m (m_+(t_i) u_i^{1-\nu})^{1/\nu}
$$
where 
$ \mu_{\bt | \bs, \bv} = \bar{\Omega}_{\bt \bs} \bar{\Omega}_\bs^{-1} \bv^\circ$,
$\Omega_{\bt | \bs, \bv} 
= \frac{ Q_{\bar{\Omega}_{\bs}} (\bv^\circ) }{\nu_{\bt | \bs, \bv}} \tilde{\Omega}$,
$Q_{\bar{\Omega}_{\bs}} (\bv^\circ) = \bv^{\circ\top} \bar{\Omega}_\bs^{-1} \bv^\circ$, 
$\tilde{\Omega} 
= \bar{\Omega}_\bt - \bar{\Omega}_{\bt \bs} \bar{\Omega}_\bs^{-1} \bar{\Omega}_{\bs \bt}$,
$\alpha_{\bt | \bs, \bv} = \tilde{\omega} \alpha_\bt$,
$\tilde{\omega} = \diag(\tilde{\Omega})^{1/2}$,
$\tau_{\bt | \bs, \bv}  = 
(\alpha_\bs + \bar{\Omega}_\bs^{-1} \bar{\Omega}_{\bs \bt} \alpha_\bt )^\top 
\bv^\circ (\nu+d)^{1/2} Q_{\bar{\Omega}_{\bs}} (\bv^\circ)^{-1/2}$,
$\kappa_{\bt | \bs, \bv} = -\tau_{(\bt, \bs)}$, 
$\nu_{\bt | \bs, \bv} = \nu + d$, 
$\bv^\circ = (\bv m_+(\bs))^{1/\nu}$ and $\bu^\circ = (\bu m_+(\bt))^{1/\nu}$.
\end{prop}
\begin{proof}
See Appendix~\ref{app:cond_int_ext_st}
\end{proof}
From Proposition~\ref{prop:cond_intensity_ext_st}, notice that the conditional 
intensity function of the extremal skew-$t$ model is the density of $T(\bs)^\nu / m_+(\bs)$, 
where $T$ is a non-central extended skew-$t$ process with parameters: 
$
\mu_{\bt | \bs, \bv}, \Omega_{\bt | \bs, \bv}, \alpha_{\bt | \bs, \bv}, 
\tau_{\bt | \bs, \bv}, \kappa_{\bt | \bs, \bv}
$ 
and 
$
\nu_{\bt | \bs, \bv},
$
which is closely related to Proposition~\ref{prop:p0_ext_st}.
%
%
% SUBSECTION INFERENCE
%
%

\subsection{Inference}
\label{ssec:inference}

Composite likelihood methods are the main strategies to bypass the computational limitations of the full likelihood approach. In particular \citet{padoan2010} and \citet{sang2014} considered the weighted composite likelihood, for which the $j$-th order is defined by
\begin{equation}
\label{eq:CL}
\mathrm{CL}_j (\bz; \theta) = \prod_{q \in \mathcal{Q}_d^{(j)}} 
\left(
\exp \{ -V(\bz_q; \theta) \}
\times \sum_{\Pi \in \mathcal{P}_q} \prod_{k=1}^{|\Pi|} - V_{\pi_k}(\bz_q; \theta) 
\right)^{w_q},
\end{equation}
for some weights $w_q\geq0$,
where $\mathcal{Q}_d^{(j)}$ represents the set of all possible subsets of size $j$ of 
$\{ 1, \ldots, d\}$ and $\bz_q$ is a $j$-dimensional subvector of 
$\bz \in \real^d_+$, $\mathcal{P}_q$ is the set of all possible partitions of $q$ where each 
partition $\Pi \in \mathcal{P}_d$ has elements $\pi_k$, for $k=1, \ldots, |\Pi|$, and $V_{\pi_k}(\cdot)$ represents
the partial derivatives of $V(\cdot)$ w.r.t $\pi_k$. \citet{wang2011} call the partition $\Pi$ the hitting scenario since it defines clusters of variables or spatial locations whose maxima comes from the same event.
The estimated parameter vector $\hat{\theta}$ maximising \eqref{eq:CL}
 can be shown to be consistent and asymptotically normally distributed \citep[see][]{padoan2010}. If the margins are jointly estimated with the dependence parameters, then the support of the parameter space depends on the parameters, which might lead to identifiability issues. However standard maximum likelihood asymptotics are still available under most practical modelling situations \citep{smith1985}.

\citet{stephenson2005} consider a different approach which relies on the knowledge of time
occurrences of each block maxima. This means that for the $n$-th block, say $\bz^n$, an 
observed partition $\Pi^n$ is associated with it, and the likelihood is then given by 
\begin{equation}
\label{eq:STlikelihood}
\mathrm{ST} (\bz; \theta) = \exp \left\{ - V(\bz; \theta) \right\}
\times \prod_{k=1}^{|\Pi|} - V_{\pi_k}(\bz; \theta).
\end{equation}

In order to compute either of the likelihoods presented above it is required to be able to 
compute partial derivatives of the exponent function $V$ up to the $d$-th order.
\citet{wadsworth2014} stress that the conditional intensity function \eqref{eq:def_cond_int} of 
$\{ Z(s_{m+1}), \ldots, Z(s_d) \}$ given $\{ Z(s_1)=z_1, \ldots, Z(s_m)=z_m \}$, i.e.~$\lambda_{\bs_{m+1:d}|\bs_{1:m}, \bz_{1:m} } (\bz_{m+1:d})$, is equivalent to
$$
\frac{-V_{1:d}(\bz)}{-V_{1:m}(\bz_{1:m}, \infty 1_{d-m})},
$$
where $\bba_{b:c} = (a_b, \ldots, a_c)$, and the denominator denotes the $m$-dimensional 
marginal intensity $\lambda_{\bs_{1:m} } (\bz_{1:m})$. Hence the partial derivatives 
$V_{1:m}(\bz)$ are obtained by integrating the conditional intensity w.r.t.~$\bz_{m+1:d}$ and 
then multiplying by $-V_{1:m}(\bz_{1:m}, \infty 1_{d-m})$.
\citet{wadsworth2014} give the partial derivatives of the $V$ function for the 
Brown-Resnick process 
while \citet{castruccio2016} also provide these results for the logistic and Reich-Shaby \citep{reich2012} models.

\begin{prop}
\label{prop:vd_ext_st}
Consider the extremal skew-$t$ model. The partial derivatives of the $V$ function 
\eqref{eq:V_ext_st} are given as follows
\begin{align*}
- V_{1:m}(\bz) &= \Psi_{d-m} \left( \bz_{m+1:d}^\circ;  
\mu_c, \Omega_c, \alpha_c, \tau_c, \kappa_c, \nu_c \right) \\
& \quad \times
\frac{
2^{(\nu-2)/2} \nu^{-m+1}   
\Gamma \left( \frac{m+\nu}{2} \right) 
\Psi \left( \tilde{\alpha}_{1:m} \sqrt{m+\nu}; -\tau^*_{1:m}, m+\nu \right) 
\prod_{i=1}^m \left( m_{i+} z_i^{1-\nu} \right)^{1/\nu}
}
{ 
\pi^{m/2} |\bar{\Omega}_\bs|^{1/2} 
\left( \bz_{1:m}^{\circ\top} \bar{\Omega}_{1:m}^{-1} \bz_{1:m}^\circ \right)^{(m+\nu)/2} 
\Phi(\tau^*_{1:m} \{ 1 + Q_{\bar{\Omega}_{1:m}}(\alpha^*_{1:m}) \}^{-1/2} )
},
\end{align*}
where $\Psi(\cdot; \kappa, \nu)$ denotes the univariate $t$ cdf with non-centrality parameter $\kappa$ and $\nu$ degrees of freedom, $\bz_{1:m}^\circ = (\bz_{1:m} m_+(\bs_{1:m}))^{1/\nu} \in \real^m$, the index $c$ 
represents $(m+1:d) | (1:m)$, where the parameters 
$\mu_c, \Omega_c, \alpha_c, \tau_c, \kappa_c, \nu_c$ are defined as in 
Proposition~\ref{prop:cond_intensity_ext_st} and 
$\tilde{\alpha}_{1:m} = \alpha_{1:m}^{*\top} \bz_{1:m}^\circ 
Q_{\bar{\Omega}_{1:m}} (\bz_{1:m}^\circ)^{-1/2} \in \real$
with $\alpha_{1:m}^* \in \real^m$ and $\tau^*_{1:m} \in \real$ respectively the $m$-dimensional 
marginal slant and extension parameter.
\end{prop}
\begin{proof} See Appendix~\ref{app:prop_vd_ext_st}. \end{proof}
%

%%%%%%%%%%%%%%%%%%%%%%%%%%%%%%%%%%%%%%%%%%%%%
%
% SECTION 3
%
%%%%%%%%%%%%%%%%%%%%%%%%%%%%%%%%%%%%%%%%%%%%%
%
\section{Simulation results}
\label{sec:sim}

In this section we use the results given in Section \ref{sec:msp} to perform an intensive simulation study for the extremal-$t$ and the extremal skew-$t$ model. Inference for these models has received little attention and thus this Section aims at quantifying the improvements associated with the use of higher order CLs than the traditional $2$-wise (pairs) and $3$-wise (triplets). We also present some strategies to allow for the use of high dimensional models (here up to $d=100$) at a reasonable computational cost.

\subsection{Simulation design}

In the following we generate 50 independent temporal replicates (say annual maxima) from the extremal-$t$ and the extremal skew-$t$ max-stable models at $d$ locations uniformly generated over the region $\cS = [-5,5] \times [-5,5]$ using Proposition~\ref{prop:p0_ext_st} with unit variance and power exponential correlation function $\rho(h) = \exp \{ - (\| h \| / r)^{\eta} \}$, $r>0, 0 < \eta \leq 2$, where $\| \cdot \|$ is the $L_2$ norm. The different correlation functions considered are represented in Figure~\ref{fig:sim_setup} (right panel): three smoothness scenarios $\eta = 1, 1.5$ and $1.95$ (solid, dashed and dotted lines) and three levels of spatial dependence by setting the range parameter to $r = 1.5, 3$ and $4.5$ (red, green and blue colours). 

When working with the student-$t$ distribution, it is known that estimates of $\nu$ are sometimes very large and variable  \citep{fernandez1999}. By extension, this phenomenom is also present when working with the non-central extended skew-$t$ distribution. \citet{davison2012b} have stressed that simultaneously estimating $r, \eta$ and $\nu$ is a delicate task, and \citet{huser2016b} have noted that the degree of freedom $\nu$ of the extremal-$t$ is difficult to estimate.
Thus $\nu$ tends to be fixed while focusing on the remaining parameters. Here we fix $\nu = 1$, which corresponds to the Schlather model (and its skew equivalent). 
Simultaneous estimation of $\nu$ and the remaining parameters is possible, as shown in \citet{beranger2017} for the extremal skew-$t$ and \citet{padoan2013} among others for the extremal-$t$, and we implement this strategy in Section \ref{sec:dataeg}.

The slant parameter is defined as a function of space $\cS$, i.e.~$\alpha_i \equiv \alpha(s_i) = \beta_1 s_{i1} + \beta_2 s_{i2}$ where $s_i = (s_{i1}, s_{i2}) \in \cS$, $i=1,\ldots, d$, and in this example we choose $\beta_1 = \beta_2 = 5$.  We use $\theta_j$ to denote a parameter vector obtained by maximisation of the CL$_j$ and $\theta_d$ to denote the use of the CL$_d$ (corresponding to the full likelihood).
  
\begin{figure}
\includegraphics[width=0.33\textwidth]{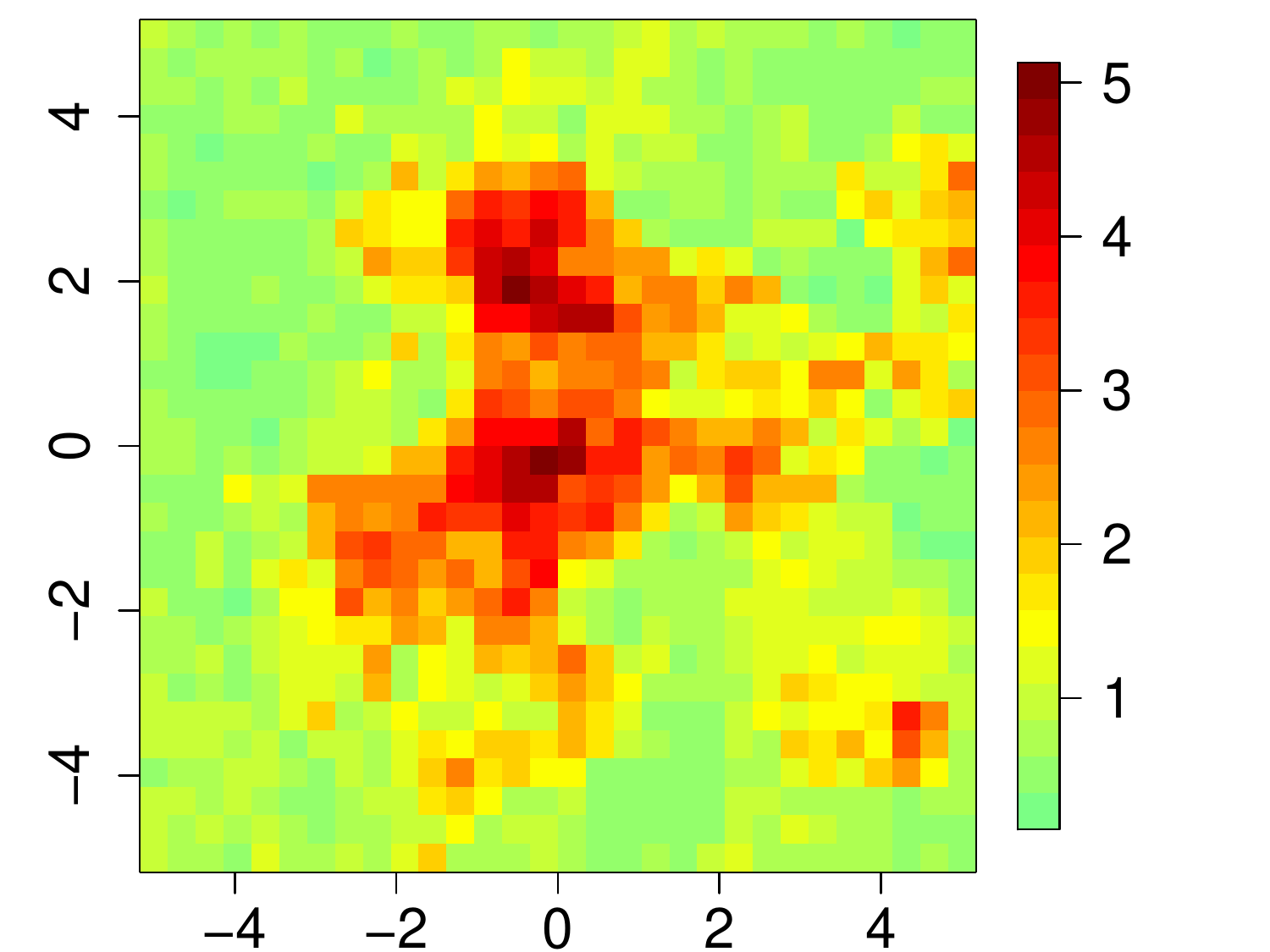}
\includegraphics[width=0.33\textwidth]{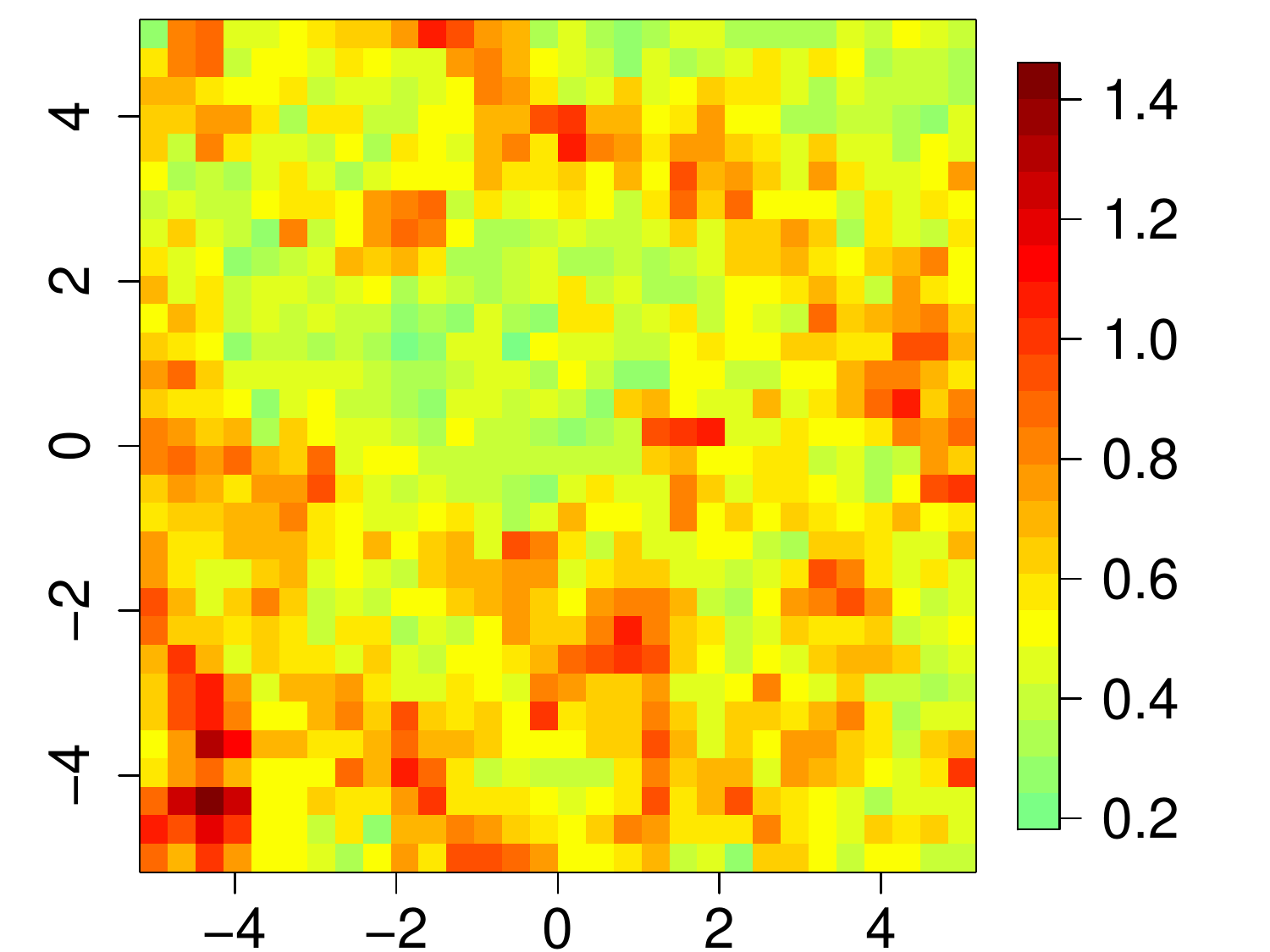}
\includegraphics[width=0.33\textwidth]{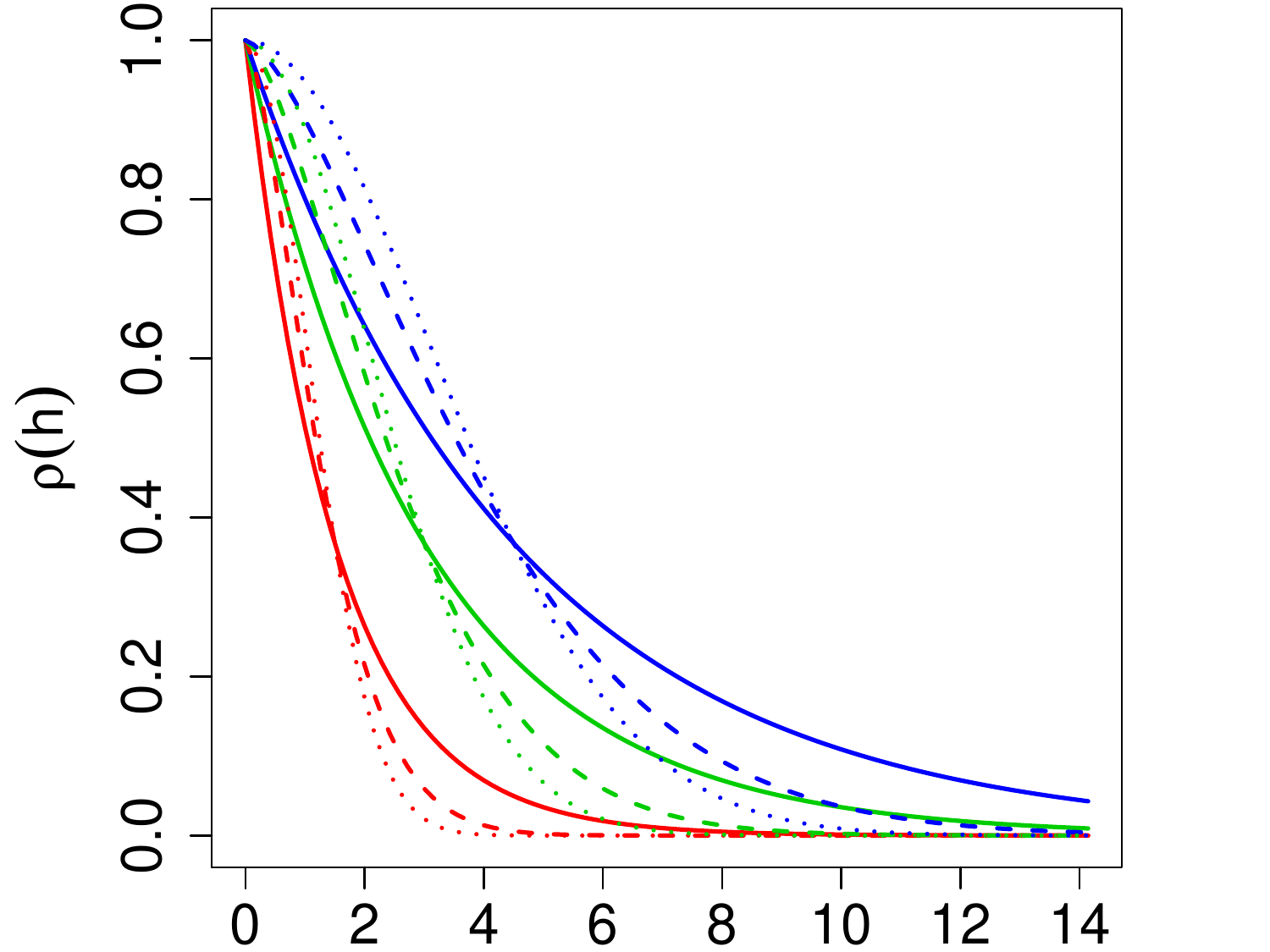}
\caption{Realisations of the extremal-$t$ (left) and extremal skew-$t$ (middle) models with $\eta=1$, $r=3$ and $\nu=1$. Right panel represents the power exponential correlation function $\rho(h)$ with smoothness $\eta = 1, 1.5$ and $1.95$ (solid, dashed and dotted lines) and the range $r = 1.5, 3$ and $4.5$ (red, green and blue colours). }
\label{fig:sim_setup}
\end{figure}

The algorithms used to generate the extremal-$t$ and the extremal skew-$t$ models are given in pseudo-code in Appendix \ref{app:algs}. The image plots of Figure~\ref{fig:sim_setup} illustrate a realisation from the extremal-$t$ and extremal skew-$t$ models on the region $\cS$ when $\eta=1$, $r=3$ and $\nu=1$ (left and middle panels). The extremal skew-$t$ algorithm incorporates a stochastic representation of the extended skew-$t$ distribution given in \citet{arellano2010}. For each replicate the algorithms also include simulation of the hitting scenario $\Pi$. This allows us to use the ST likelihood in equation \eqref{eq:STlikelihood}, which greatly simplifies the evaluation of both CL$_d$ and CL$_j$. An alternative approach that avoids the need to compute the exponent function for all possible partitions $\Pi$ is to treat the hitting scenario as missing and use the expectation--maximization algorithm \citep{huser2019}.

\subsection{Stephenson-Tawn likelihood evaluation}
\label{ssec:STeval}

For each observation, the log-likelihood from equation \eqref{eq:STlikelihood} includes the evaluation of $V(\bz;\theta)$ and of $\log(-V_{\pi_k}(\bz;\theta))$ for $k = 1,\dots,|\Pi|$. This respectively requires $d$ evaluations of $\Psi_{d-1}(\cdot)$ and $|\Pi|$ evaluations of the form $\log(\Psi_{d-m}(\cdot))$, where $m=|\pi_k|$. As discussed by \citet{dombry2016} the evaluation of these cdfs is a computationally difficult task even for moderate $d$, and we thus suggest to overcome this by controlling the degree of approximation of these quantities.
Extended skew-$t$ distribution functions can be written in terms of the multivariate $t$-distribution, which we evaluate using quasi-Monte Carlo approximations; see the algorithm in Section 3.2 of \citet{genz2002}. The term `quasi' refers to the fact that the Monte Carlo simulations are based on lattice points and are therefore more evenly distributed than a standard Monte Carlo algorithm. See also \citet{genz1992} and \citet{genz1993} for the evaluation of the multivariate normal distribution, as required for Brown-Resnick processes. 

The computational importance of multivariate $t$-cdf (or normal cdf) evaluation within max-stable process models has been previously discussed; see e.g.\ \citet{wadsworth2014}, \citet{castruccio2016}, \citet{thibaud2016} and \citet{defondeville2018}. 
For instance, \citet{thibaud2016} consider Monte Carlo estimates of the Gaussian cdf involved in the Brown-Resnick model while \citet{defondeville2018} show that quasi-Monte Carlo methods produce faster convergence rates than classical Monte Carlo estimates. \citet{defondeville2018} also further reduce the computation time by using randomly shifted deterministic lattice rules to compute multivariate normal distribution functions and argue that their methodology can be extended to the extremal-$t$ model.

The approach we use is as follows. The original algorithm in \citet{genz2002} controls the absolute error (we use $\epsilon = 0.001$). It is important to adjust the algorithm, controlling the error on a log-scale to take account of the fact that the logarithm of $\Psi_{d-m}(\cdot)$ is required. Fewer Monte Carlo simulations are then needed. The evaluations of $\Psi_{d-m}(\cdot)$ in $V_{\pi_k}(\bz;\theta)$ are also relatively more important than those of $\Psi_{d-1}(\cdot)$ in $V(\bz;\theta)$. The algorithm parameters $N_{min}$ and $N_{max}$ control the minimum and maximum number of simulations used. The maximum number is used only if the approximation error remains above $\epsilon$.  

\begin{table}[!ht]
\centering
\footnotesize
$  
\begin{array}{c c c c c c c c}
\mathrm{j} & 2 & 3 & 4 & 5 & 10 & d\mbox{ (Type I)} & d\mbox{ (Type II)} \\
\hline
 \Psi_{j-m}(\cdot) & 100,1000 & 100,1000 & 50,500 & 50,500 & \multicolumn{1}{c|}{20,200} & 50,500 & 20,200 \\
 \Psi_{j-1}(\cdot) & 10,100 & 10,100 & 5,50 & 5,50 & \multicolumn{1}{c|}{2,20} & 5,50 & 2,20 \\
\hline
\end{array}
$
\caption{Number of quasi-Monte Carlo simulations $N_{min},N_{max}$ to compute each $\Psi_{j-m}(\cdot)$ and $\Psi_{j-1}(\cdot)$ terms in $V_{\pi_k}(\bz;\theta)$ for each $j$-wise composite likelihood. The case $j=d$ corresponds to the full likelihood, where two different approximations are considered.
}
\label{tab:approx}
\end{table}

Table~\ref{tab:approx} provides the different $N_{min}$ and $N_{max}$ levels considered in each $j$-wise CL estimation. For the $2$-wise and $3$-wise CL estimations of the extremal-$t$ model, the approximation error almost always reduces below $\epsilon$ before reaching $N_{max}$. The case $j=d$ corresponds to the full likelihood, where two different approximations (Type I and Type II) are considered. These are obtained by varying the minimum and maximum number of quasi-Monte Carlo simulations. Type I is the same level as when $j=4$ or $5$ whereas Type II is a rougher approximation also used when $j=10$. Section \ref{ssec:fulldimapprox} discusses ST likelihood approximations for the cases $d=20$, $d=50$, and $d=100$.  

Likelihood evaluations for the extremal skew-$t$ are obviously slower than the simpler extremal-$t$ because the former requires the evaluation of the multivariate extended skew-$t$ cdf. Likelihood evaluations can easily be parallelized over the number of observations. Every likelihood evaluation conducted in this section was evaluated in parallel using $16$ CPUs. An alternative would be use $d$ CPUs to parallelize the $V(\bz;\theta)$ and $\log(-V_{\pi_k}(\bz;\theta))$ evaluations directly, perhaps combined with vectorised operations \citep{warne+sd19}. It would even be possible to do both, if a large enough number of CPUs were available.

\subsection{Approximation of the Stephenson-Tawn likelihood}
\label{ssec:fulldimapprox}

We first investigate the effect of the cdf approximations on the parameter estimates obtained using the full (i.e.\ non composite) likelihood.
Define the root mean square error (RMSE) of an estimator $\hat{\theta}$ of $\theta$, over 500 replicates, by $\textrm{RMSE}(\hat{\theta}) = \sqrt{\textrm{b}(\hat{\theta})^2 + \textrm{sd}(\hat{\theta})^2}$, where the bias is $\textrm{b}(\hat{\theta}) = \bar{\hat{\theta}} - \theta$, $\bar{\hat{\theta}} = \sum_{i=1}^{500} \hat{\theta}_i / 500$, and the standard deviation $\textrm{sd}(\hat{\theta}) = \sqrt{ \sum_{i=1}^{500} (\hat{\theta}_i - \bar{\hat{\theta}})^2 / 499}$. 
Table~\ref{tab:Ext_full_RMSE} provides a comparison of the RMSEs of the smooth and range estimators (respectively $\hat{\eta}$ and $\hat{r}$) for the extremal-$t$ and extremal skew-$t$ models obtained using Type I and Type II approximations of the cdf terms. Table~\ref{tab:Ext_full_BIAS} in Appendix~\ref{app:sim} provides corresponding bias estimates. 

As expected, a larger number of sites (from $20$ to $100$) and better approximations (Type I rather than Type II) yield smaller RMSEs. For fixed smoothness and dimension, as the process becomes more spread ($r$ large), the RMSE of the range estimator $\hat{r}$ tends to increase. This might be explained by the difficulty to dissociate independent site locations. When the range $r$ is fixed, as the process becomes smoother ($\eta$ large) the smooth estimator $\hat{\eta}$ becomes more accurate.
The RMSEs of the smoothness and range parameters are larger in the four parameter model due to the additional model complexity.

The estimation of $\beta_1$ and $\beta_2$, which define the $d$-dimensional slant parameter vector as a function of space, is less robust. Maximum likelihood estimation methods for skewed distributions often yield similar robustness and identifiability issues. Our simulation revealed the presence of a very few abnormally large skewness parameters, impacting the RMSE values. 

Overall, Table~\ref{tab:Ext_full_RMSE} indicates that the ST likelihood yields accurate estimates for the model parameters. 
\citet{wadsworth2014} and \citet{huser2016} have highlighted the presence of bias which increases with the dimension $d$ and under weaker dependence when the hitting scenario does not come from the limiting distribution of scenarios. In the above analysis the data are simulated exactly from the process with hitting scenarios (see Appendix~\ref{app:algs}), and it is thus unsurprising to observe biases that are small and decreasing with dimension $d$ for the parameter estimates of both the extremal-$t$ and extremal skew-$t$ models (see Table~\ref{tab:Ext_full_BIAS} in Appendix~\ref{app:sim}).
	
\clearpage
\begin{landscape}
	\begin{table}
		\centering
		\footnotesize
		$
		\begin{array}{ccc c cccccccc c cccccccccccccc}
		\multicolumn{4}{c}{} & \multicolumn{8}{c}{\textrm{extremal-$t$}} & & \multicolumn{14}{c}{\textrm{extremal skew-$t$}} \\
		\cline{5-12} \cline{14-27} 
		\multicolumn{4}{c}{} & \multicolumn{2}{c}{r=1.5} & & \multicolumn{2}{c}{r=3.0} & & \multicolumn{2}{c}{r=4.5} & & \multicolumn{4}{c}{r=1.5} & & \multicolumn{4}{c}{r=3.0} & & \multicolumn{4}{c}{r=4.5} \\
		\cline{5-6} \cline{8-9} \cline{11-12} \cline{14-17} \cline{19-22} \cline{24-27}
		& & \mbox{Type} &  & \hat{\eta}_j & \hat{r}_j & & \hat{\eta}_j & \hat{r}_j & & \hat{\eta}_j & \hat{r}_j & & \hat{\eta}_j & \hat{r}_j & \hat{\beta}_{1j} & \hat{\beta}_{2j} & & \hat{\eta}_j & \hat{r}_j & \hat{\beta}_{1j} & \hat{\beta}_{2j} & & \hat{\eta}_j & \hat{r}_j & \hat{\beta}_{1j} & \hat{\beta}_{2j} \\
		\hline
		d=20 & \eta=1.00 & \mbox{I} & & 0.058 & 0.125 &  & 0.055 & 0.260 &  & 0.047 & 0.419 & & 0.094 & 0.164 & 0.782 & 0.677 &  & 0.081 & 0.396 & 0.559 & 0.448 &  & 0.054 & 0.466 & 0.310 & 0.284 \\ 
		& & \mbox{II} & & 0.061 & 0.124 &  & 0.048 & 0.246 &  & 0.044 & 0.378 & & 0.106 & 0.223 & 0.380 & 0.394 &  & 0.117 & 0.478 & 0.633 & 0.893 &  & 0.101 & 0.670 & 0.831 & 0.883 \\
		& \eta=1.50 & \mbox{I} & & 0.046 & 0.076 &  & 0.036 & 0.164 &  & 0.030 & 0.232 &  & 0.114 & 0.122 & 1.702 & 0.899 &  & 0.046 & 0.238 & 0.406 & 0.395 &  & 0.036 & 0.450 & 0.429 & 0.427 \\ 
		& & \mbox{II} & & 0.045 & 0.074 &  & 0.032 & 0.156 &  & 0.027 & 0.226 &  & 0.096 & 0.142 & 0.363 & 0.359 &  & 0.050 & 0.236 & 0.324 & 0.323 &  & 0.039 & 0.345 & 0.413 & 0.401 \\
		& \eta=1.95 & \mbox{I} & & 0.025 & 0.051 &  & 0.008 & 0.077 &  & 0.005 & 0.112 &  & 0.024 & 0.070 & 1.265 & 0.900 &  & 0.011 & 0.139 & 0.420 & 0.420 &  & 0.006 & 0.185 & 0.263 & 0.241 \\
		& & \mbox{II} & & 0.018 & 0.049 &  & 0.009 & 0.087 &  & 0.005 & 0.116 &  & 0.028 & 0.090 & 0.529 & 0.476 &  & 0.019 & 0.129 & 0.532 & 0.609 &  & 0.008 & 0.210 & 0.367 & 0.436 \\
		\hline
		d=50 & \eta=1.00 & \mbox{I} & & 0.024 & 0.057 &  & 0.022 & 0.137 &  & 0.018 & 0.207 &  & 0.045 & 0.120 & 0.175 & 0.119 &  & 0.034 & 0.211 & 0.216 & 0.176 &  & 0.032 & 0.364 & 0.220 & 0.216 \\ 
		& & \mbox{II} & & 0.051 & 0.089 &  & 0.022 & 0.157 &  & 0.025 & 0.253 &  & 0.051 & 0.152 & 0.214 & 0.216 &  & 0.042 & 0.266 & 0.189 & 0.196 &  & 0.045 & 0.405 & 0.218 & 0.264 \\
		& \eta=1.50 & \mbox{I} & & 0.012 & 0.039 &  & 0.013 & 0.095 &  & 0.013 & 0.139 &  & 0.031 & 0.068 & 0.118 & 0.111 &  & 0.024 & 0.190 & 0.112 & 0.104 &  & 0.021 & 0.229 & 0.156 & 0.153 \\
		& & \mbox{II} & & 0.020 & 0.046 &  & 0.016 & 0.096 &  & 0.014 & 0.153 &  & 0.038 & 0.098 & 0.183 & 0.169 &  & 0.029 & 0.185 & 0.145 & 0.349 &  & 0.028 & 0.277 & 0.304 & 0.307 \\
		& \eta=1.95 & \mbox{I} & & 0.004 & 0.023 &  & 0.002 & 0.049 &  & 0.002 & 0.086&  & 0.008 & 0.045 & 0.145 & 0.137 &  & 0.003 & 0.081 & 0.215 & 0.214 &  & 0.003 & 0.111 & 0.148 & 0.168 \\
		& & \mbox{II} & & 0.005 & 0.029 &  & 0.003 & 0.062 &  & 0.002 & 0.101 &  & 0.010 & 0.046 & 1.293 & 1.267 &  & 0.004 & 0.095 & 0.282 & 0.269 &  & 0.004 & 0.158 & 0.202 & 0.261 \\
		\hline
		d=100 & \eta=1.00 & \mbox{I} & & 0.020 & 0.052 &  & 0.017 & 0.128 &  & 0.015 & 0.195 &  & 0.044 & 0.110 & 0.098 & 0.106 &  & 0.031 & 0.203 & 0.090 & 0.085 &  & 0.028 & 0.337 & 0.055 & 0.059 \\
		& & \mbox{II} & & 0.025 & 0.077 &  & 0.021 & 0.189 &  & 0.022 & 0.277 &  & 0.045 & 0.116 & 0.140 & 0.141 &  & 0.035 & 0.312 & 0.111 & 0.131 &  & 0.040 & 0.369 & 0.219 & 0.215 \\  
		& \eta=1.50 & \mbox{I} & & 0.011 & 0.028 &  & 0.01 & 0.068 &  & 0.011 & 0.127 &  & 0.021 & 0.053 & 0.071 & 0.068 &  & 0.019 & 0.122 & 0.051 & 0.045 &  & 0.024 & 0.279 & 0.239 & 0.255 \\
		& & \mbox{II} & & 0.013 & 0.039 &  & 0.012 & 0.091 &  & 0.014 & 0.155 &  & 0.028 & 0.077 & 0.072 & 0.071 &  & 0.034 & 0.272 & 0.203 & 0.227 &  & 0.024 & 0.301 & 0.341 & 0.355 \\
		& \eta=1.95 & \mbox{I} & & 0.002 & 0.017 &  & 0.001 & 0.040 &  & 0.002 & 0.067 &  & 0.003 & 0.029 & 0.547 & 0.506 &  & 0.002 & 0.072 & 0.070 & 0.059 &  & 0.004 & 0.143 & 0.553 & 0.600 \\
		& & \mbox{II} & & 0.002 & 0.023 &  & 0.002 & 0.052 &  & 0.002 & 0.099 &  & 0.004 & 0.035 & 0.103 & 0.102 &  & 0.004 & 0.102 & 0.274 & 0.274 &  & 0.005 & 0.169 & 0.698 & 0.743 \\
		\hline

		\end{array}
		$
		\caption{\small RMSEs for $\hat{\theta}_j = (\hat{\eta}_j, \hat{r}_j)$ and $\hat{\theta}_j = (\hat{\eta}_j, \hat{r}_j, \hat{\beta}_{1j}, \hat{\beta}_{2j})$ the parameter vectors of the extremal-$t$ and extremal skew-$t$ models, using the full likelihood Type I and Type II approximations given in Table~\ref{tab:approx} when $d=20, 50$ and $100$ sites are considered. Calculations are based on 500 replicate maximisations.}
		\label{tab:Ext_full_RMSE}
	\end{table}
\end{landscape}

\begin{figure}[t!]
	\centering
	$
	\begin{array}{ccc}
	\includegraphics[width=0.4\linewidth]{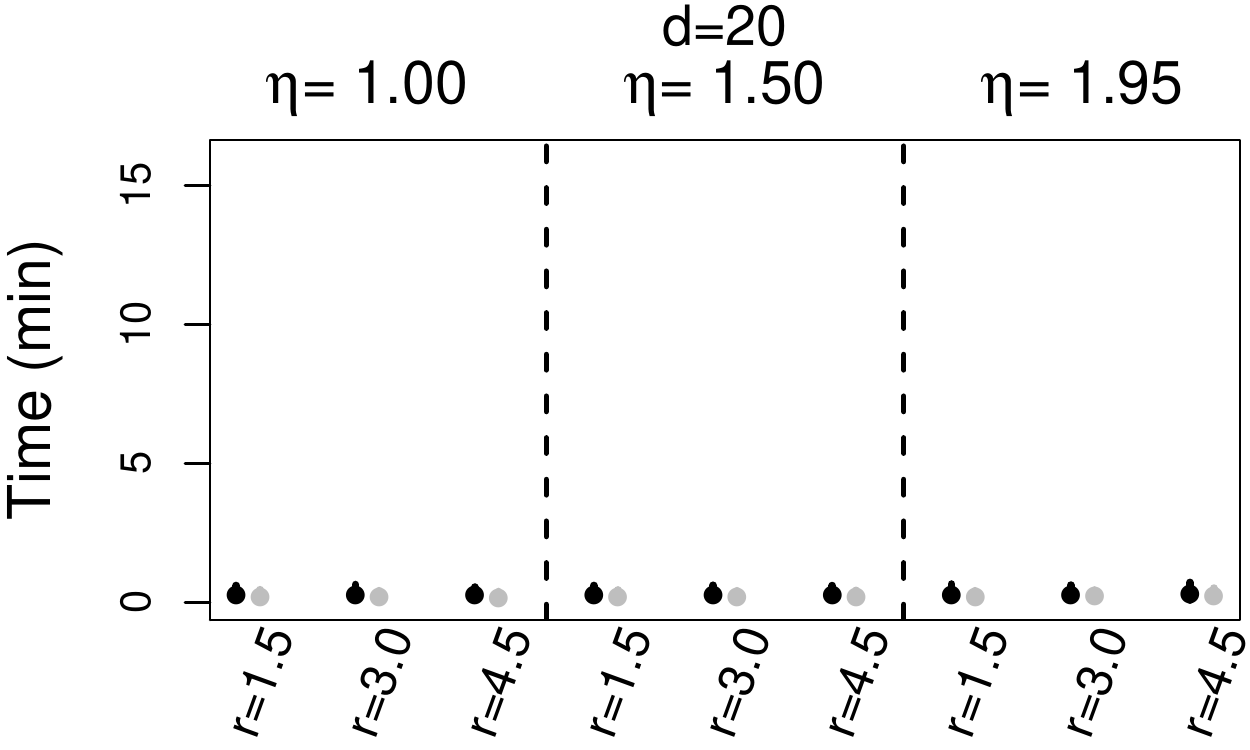} & &
	\includegraphics[width=0.4\linewidth]{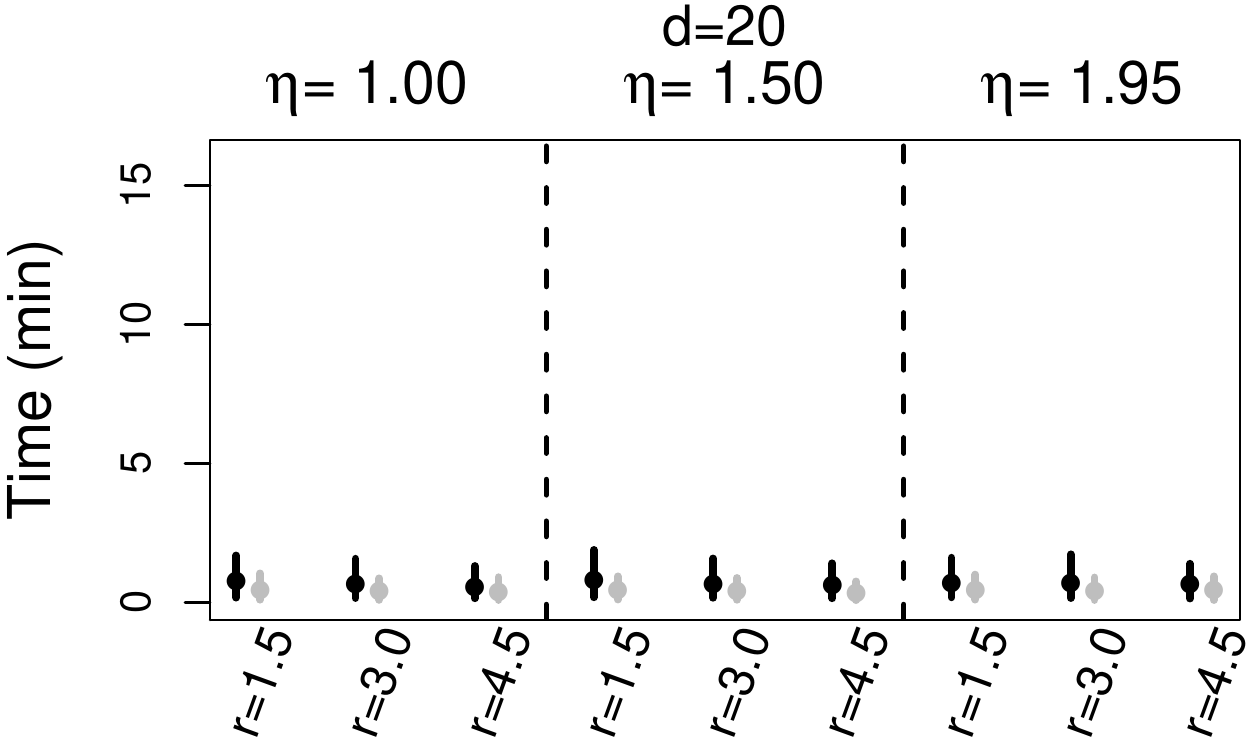} \\
	\includegraphics[width=0.4\linewidth]{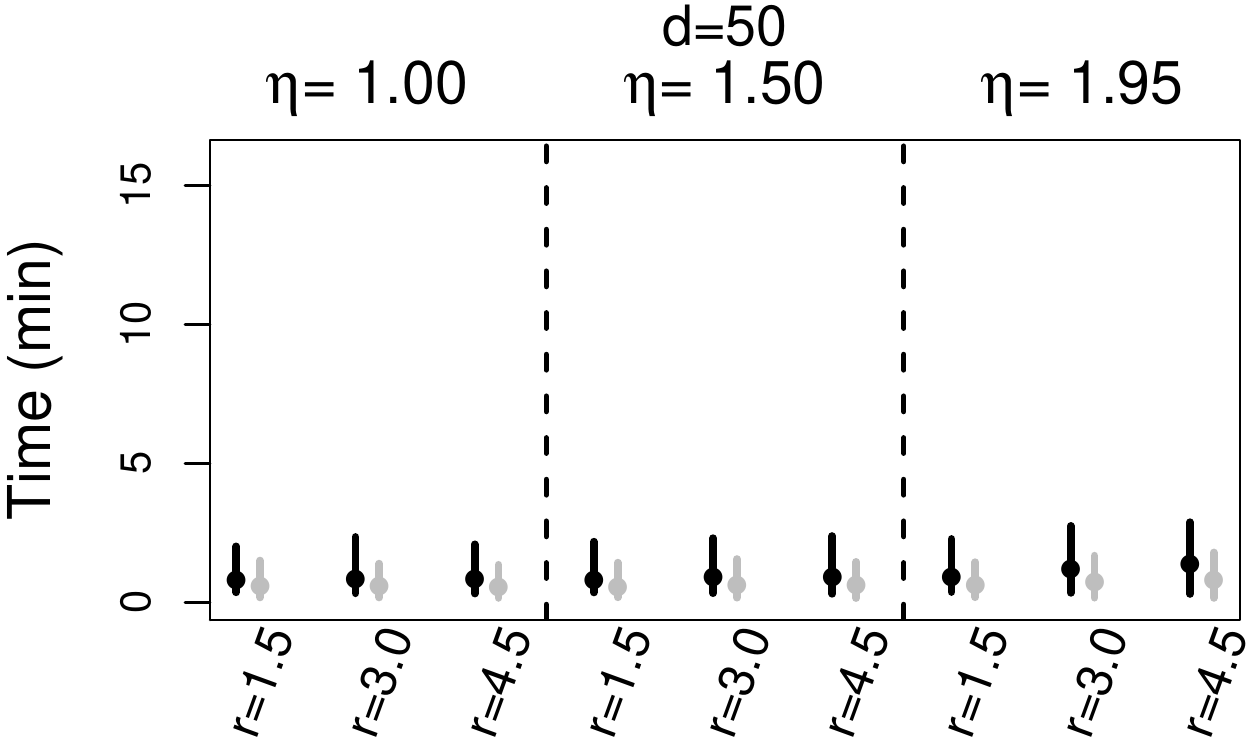} & &
	\includegraphics[width=0.4\linewidth]{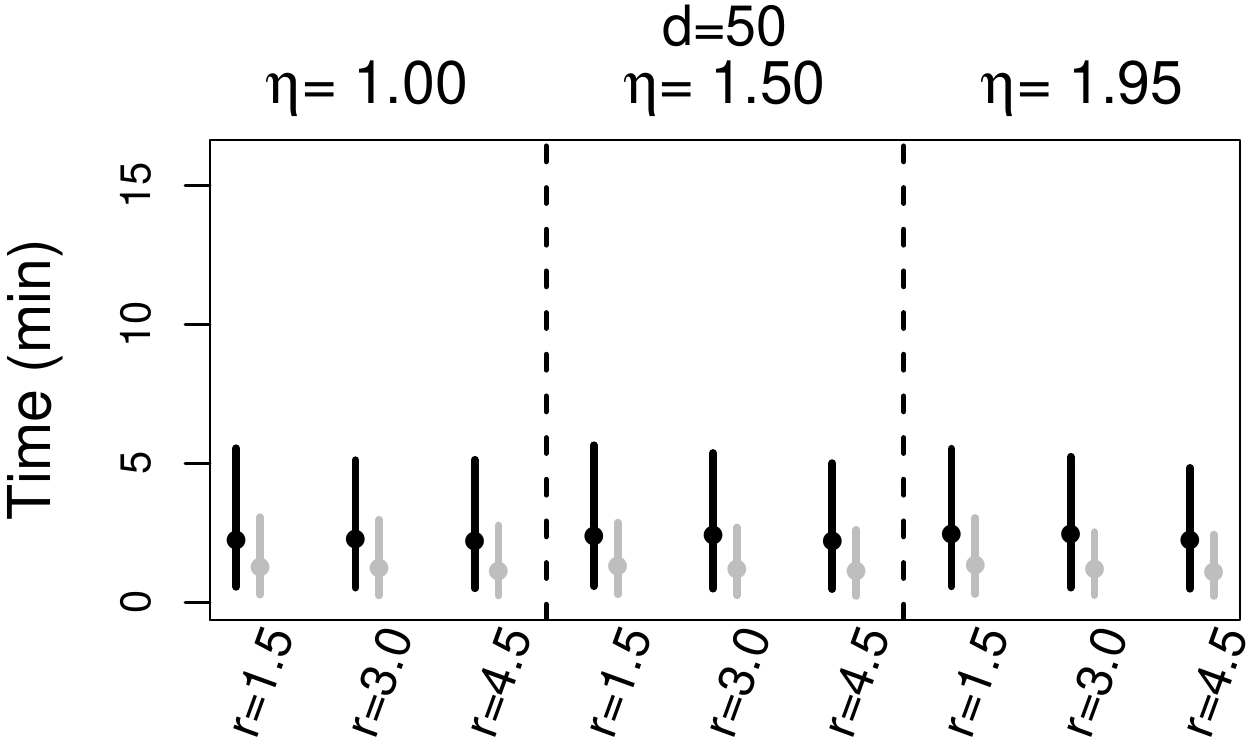} \\
	\includegraphics[width=0.4\linewidth]{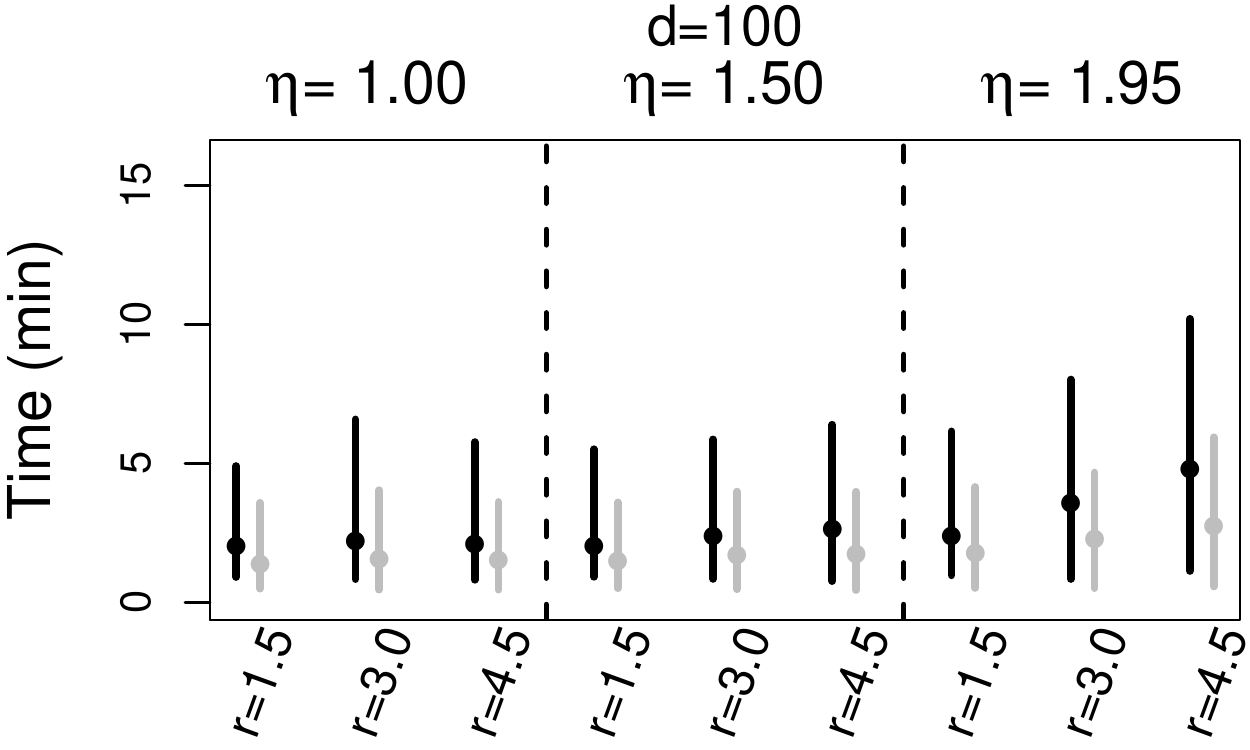} & &
	\includegraphics[width=0.4\linewidth]{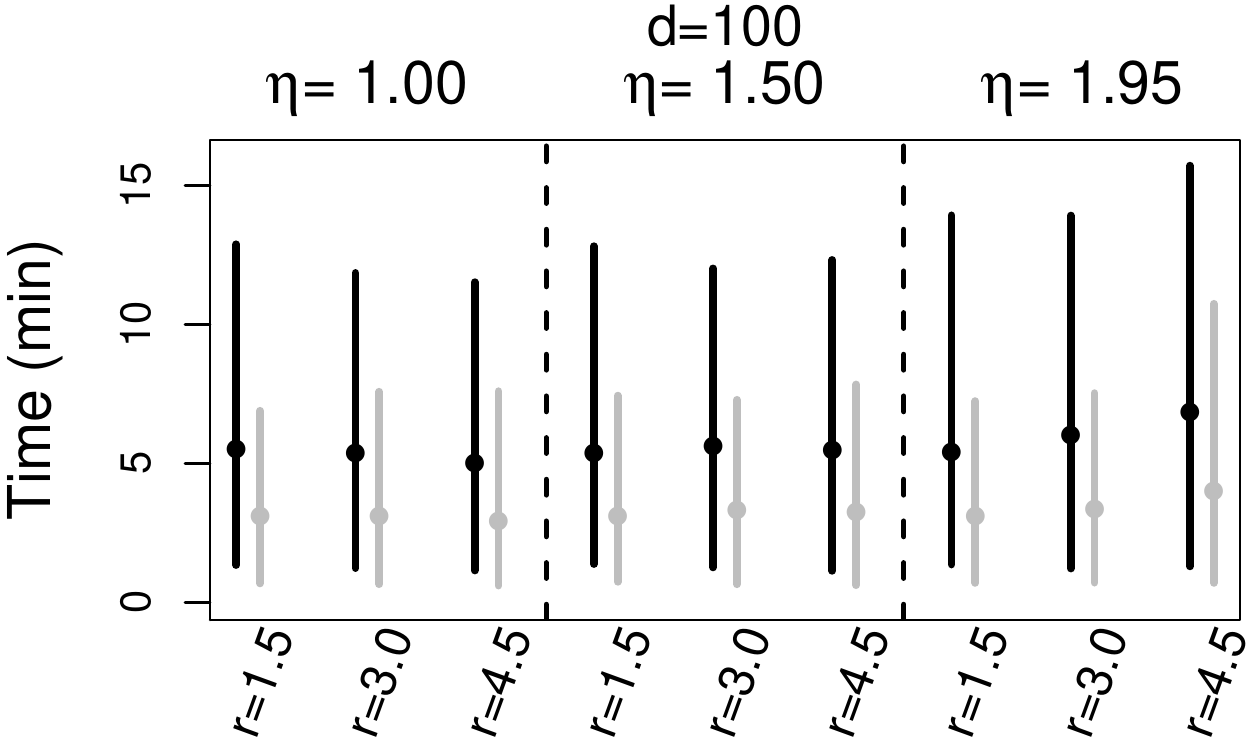} 
	\end{array}
	$
	\caption{Mean time (in minutes) and $95\%$ confidence region for the maximisation of the extremal-$t$ (left) and extremal skew-$t$ (right) full likelihood function, using the Type I (black) and Type II (grey) approximations as given in Table~\ref{tab:approx}, for $d=20, 50$ and $100$ (top to bottom).}
	\label{fig:time_full}
\end{figure}

Figure~\ref{fig:time_full} gives computational times under the same settings as Table~\ref{tab:Ext_full_RMSE}. It illustrates the mean time and its $95\%$ confidence region to maximise the likelihood function in parallel over 16 CPUs, using the Type I (black) and Type II (grey) approximations. As expected, the computation time is lower when using the rougher approximation (Type II) and increases with the number of sites. Focusing on the extremal skew-$t$, using approximation Type I, the median maximisation times are around one, three and six
minutes respectively for $d=20$, $d=50$ and $d=100$. These values are relatively constant across the different dependence structures and can be halved by using the Type II approximation. Neither the smoothness nor the spread of the processes has a noticeable impact on the speed. 

The computational speed is fast due to both the approximations and the use of the ST likelihood. For comparison, \citet{castruccio2016} stated that full likelihood estimation is limited to $d=12$ or $13$, and  a single iteration of the expectation-maximisation algorithm of \citet{huser2019} takes several hours for the Brown-Resnick model. 

For a large number of sites it seems favourable to use the rougher Type II approximation, which produces substantial gains in computation time at a relatively low accuracy loss. This seems a particularly appealing strategy for complex high-dimensional models.

\subsection{Composite $j$-wise likelihoods}

We now investigate the performance of various high-order composite likelihoods using the full likelihood as reference.
The composite likelihood defined in \eqref{eq:CL} requires the computation of the $\binom{d}{j}$ elements in $\mathcal{Q}_d^{(j)}$. For fixed dimension $d$, even a moderately high $j$ (compared to $d$) will require higher computational cost than the full likelihood. It is reasonable to believe that the required number of elements of $\mathcal{Q}_d^{(j)}$ should decrease as $j$ increases, reducing the computational burden for similar efficiency level. We therefore set binary weights in \eqref{eq:CL} such that only a restricted number of the $\binom{d}{j}$ elements are selected. For $j \in \{ 1, \ldots, d\}$ and some $q \in \mathcal{Q}_d^{(j)}$ we define the weights as
$$
w_q = \left\{ \begin{array}{ccl}  
1 & & \mathrm{if }\, \max_{i,k \in q; i \neq k} \| s_i - s_k \| < u \\
0 & & \mathrm{otherwise}
\end{array} \right.,
$$
where $u > 0$. 

In the following we focus on $d=20$ and we consider different thresholds $u$ such that approximately $50$ tuples are used to compute each CL function. See also \cite{castruccio2016}, who compare efficiencies of CL estimators in smaller ($d=11$) dimensions. 

We evaluate both the statistical efficiency and computational cost of high-order CL estimators. Thus a comparison to the full likelihood estimator is established through a Time Root Relative Efficiency (TRRE) criterion defined as
$$ 
\mathrm{TRRE} (\theta_j) = \frac{\textrm{RMSE}(\hat{\theta}_d)}{\textrm{RMSE}(\hat{\theta}_j)}
\times  \frac{\textrm{time}(\hat{\theta}_d)}{\textrm{time}(\hat{\theta}_j)},
$$
the product of the Root Relative Efficiency (RRE) and time ratio. 
\begin{table}[t!]
$$
\scriptscriptstyle
\begin{array}{cc c ccc c ccc}
\multicolumn{3}{c}{} & \multicolumn{3}{c}{\textrm{extremal-$t$}} & & \multicolumn{3}{c}{\textrm{extremal skew-$t$}} \\
\cline{4-6} \cline{8-10} 
& & & r=1.5 & r=3.0 & r=4.5 & & r=1.5 & r=3.0 & r=4.5 \\
\hline
\eta=1.00 & j=2 & & 103/106 & 62/72 & 56/50 &  & 06/01/07/06 & 04/06/05/04 & 03/03/02/02 \\
& j=3 & & 33/35 & 32/28 & 27/23 &  & 07/02/17/12 & 09/04/12/08 & 06/05/05/04 \\
& j=4 & & 15/18 & 16/16 & 16/14 &  & 19/19/25/19 & 21/15/21/10 & 11/03/07/05 \\
& j=5 & & 10/11 & 10/10 & 09/09 &  & 07/01/11/08 & 10/04/07/04 & 07/06/04/03 \\
& j=10 & & 14/15 & 14/14 & 13/13 &  & 12/16/17/13 & 15/19/10/12 & 08/11/07/06 \\
\hline
\eta=1.50 & j=2 & & 64/77 & 61/77 & 45/56 &  & 11/03/17/09 & 05/03/04/03 & 04/03/04/03 \\
& j=3 & & 21/29 & 25/30 & 25/26 &  & 18/02/41/17 & 10/03/09/06 & 11/08/13/09 \\
& j=4 & & 11/16 & 12/17 & 13/15 &  & 37/31/82/29 & 12/12/13/08 & 10/10/12/08 \\
& j=5 & & 08/11 & 08/11 & 08/09 &  & 21/18/34/16 & 09/02/06/05 & 09/16/07/05 \\ 
& j=10 & & 12/14 & 13/15 & 12/14 &  & 28/27/30/19 & 18/26/19/14 & 15/23/13/12 \\
\hline
\eta=1.95 & j=2 & & 66/68 & 36/75 & 23/58 &  & 04/01/11/07 & 04/02/04/03 & 03/02/02/02 \\
& j=3 & & 21/31 & 19/32 & 16/29 & & 12/01/18/11 & 05/03/09/06 & 03/07/06/04 \\
& j=4 & & 13/19 & 20/25 & 18/24 & & 21/28/37/26 & 07/17/14/10 & 04/09/04/03 \\
& j=5 & & 07/12 & 14/17 & 15/17 &  & 19/19/26/16 & 14/25/11/08 & 11/16/09/07 \\
& j=10 & & 09/15 & 32/20 & 22/18 &  & 13/20/03/04 & 16/26/10/10 & 22/19/10/09 \\
\hline
\end{array}
$$
\caption{\small Time root relative efficiency (TRRE) of $\hat{\eta}_j/\hat{r}_j$ and $\hat{\eta}_j /\hat{r}_j/\hat{\beta}_{1j}/\hat{\beta}_{2j}$ for the extremal-$t$ and extremal skew-$t$ models with $\nu=1$ and $d=20$. Larger values are preferable under the TRRE criterion.}
\label{tab:TRRE}
\end{table}
Table~\ref{tab:TRRE} presents the TRRE of the parameters of the extremal-$t$ and extremal skew-$t$ models with various range and smoothness parameters and fixed degree of freedom $\nu=1$. From the TRRE of the extremal-$t$ estimates (left columns), it appears that low-dimensional composite likelihood methods give the best trade-off between accuracy and computational cost. Moreover this  becomes more pronounced as the smoothness and range parameter reduce. For the extremal skew-$t$ estimates (right columns), the TRREs show that the higher-order composite likelihoods become more efficient, with e.g.~the $4$-wise composite likelihood consistently performing well across a wide range of scenarios. 

A more detailed explanation of these results is provided by separately analysing statistical and computational efficiencies. Table~\ref{tab:RMSE_ET_SKT} in Appendix~\ref{app:sim} provides the RMSEs. 
It highlights that the RMSEs are reducing as $j$ increases, with the highest statistical efficiency obtained for $j=10$. Part of this trend is hidden by the constant number of tuples considered across each method and the increased degree of approximation as function of $j$. 

Figure~\ref{fig:Time_ET_SKT} provides the associated computational timings. Figure~\ref{fig:Time_ET_SKT} demonstrates that the lowest computation times for the extremal skew-$t$ (bottom panels) are shared by the pairwise composite likelihood ($j=2$) and the full likelihood ($j=d$). For the extremal-$t$ (top panels) the maximisation times for $j=2$ are lower than for $j=d$ and increase gradually from $j=2$ up to $j=10$. 

The same number of tuples are considered in each $j$-wise composite likelihood; due to the cdf evaluation, for fixed $N_{max}$ the mean maximisation time increases with $j$. However if the approximation becomes rougher (i.e.~if $N_{max}$ decreases for increasing $j$), the times can decrease. This confirms the utility of our strategy of controlling the degree of approximation in the exponent function and its derivatives. The drop in time between $j=10$ and $j=d$ suggests that it might also be useful to consider fewer tuples as $j$ increases. 

Table~\ref{tab:RMSE_ET_SKT}  in Appendix~\ref{app:sim} shows that pairwise and triplewise CLs can yield much larger RMSEs than higher order ($j=10$) CLs. For more flexible high-dimensional models such as the extremal skew-$t$, higher order composite likelihoods should be considered, and these require fine strategies to control the computation time.

\begin{figure}[ht]
	\includegraphics[width=\textwidth]{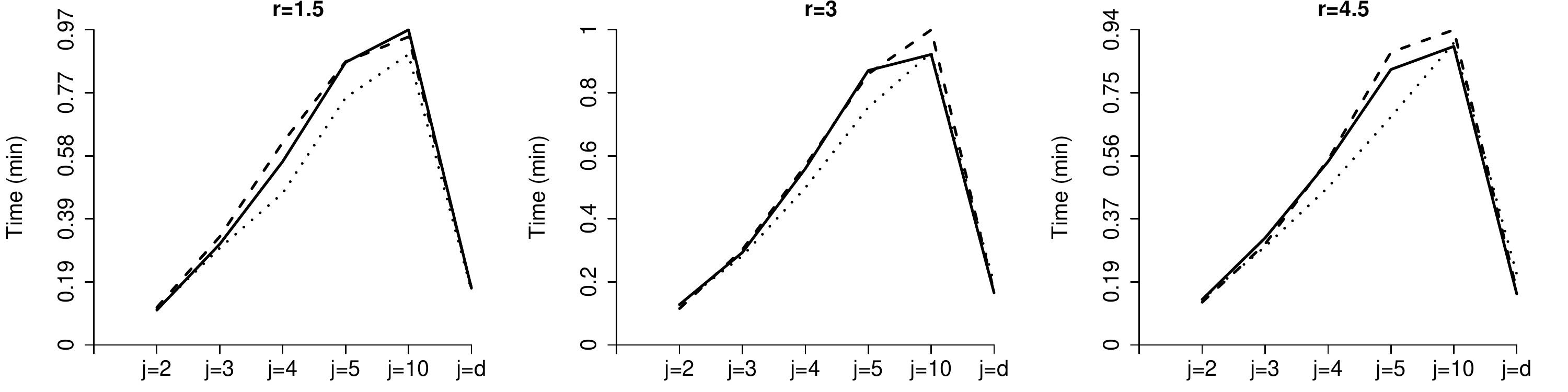} \\
	\includegraphics[width=\textwidth]{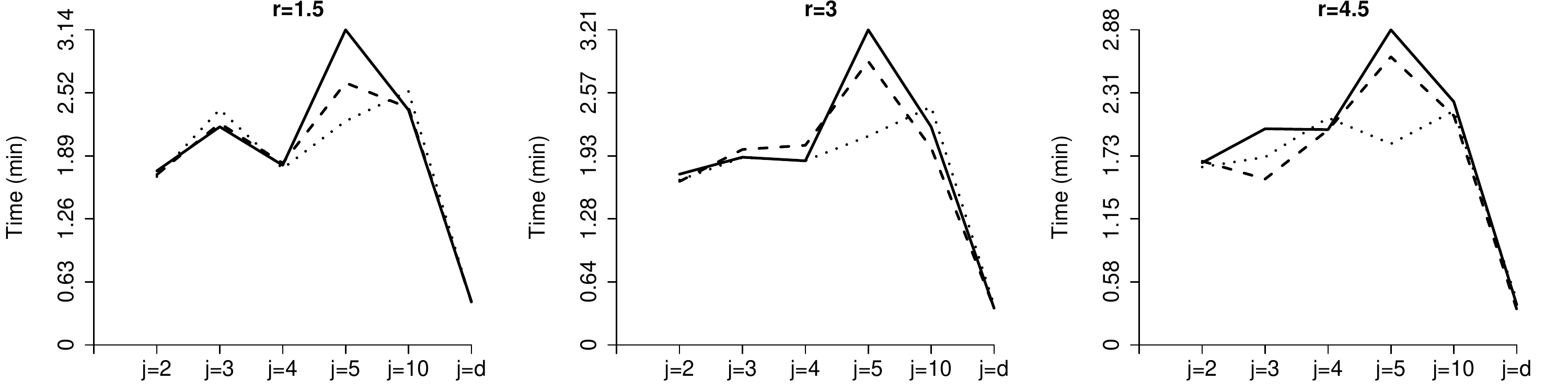}
	\caption{\small Average time (in minutes) for the maximisation of the $j$-wise composite likelihood function of the extremal-$t$ (top) and extremal skew-$t$ (bottom) for $\nu=1$ and $d=20$. Smoothness values $\eta = 1, 1.5$ and $1.95$ are represented by solid, dashed and dotted lines. The case $j=d$ corresponds to full likelihood estimation with the Type II approximation from Table~\ref{tab:approx}.}
	\label{fig:Time_ET_SKT}
\end{figure}

We examine the bias of our methodology through Table~\ref{tab:BIAS_ET_SKT} in Appendix~\ref{app:sim}. The method seems to yield large biases for low-degree CL for the extremal skew-$t$ whereas the bias is relatively small for the extremal-$t$. In general, for fixed approximation levels and the same number of tuples, increasing the order of the CL increases the bias.

%%%%%%%%%%%%%%%%%%%%%%%%%%%%%%%%%%%%%%%%%%%%
%
% DATA EXAMPLE
%
%%%%%%%%%%%%%%%%%%%%%%%%%%%%%%%%%%%%%%%%%%%%
%

\section{Temperature Data Example}
\label{sec:dataeg}

We present an illustrative analysis of the application of an extremal skew-$t$ process using temperature data around the city of Melbourne, Australia. The data is a gridded commercial product \citep{jeffrey2001} interpolated from a network of weather stations, recorded during the $N=50$ year period 1961--2010. The $d=90$ stations are on a $0.15$ degree (approximately $13$ kilometre) grid in a $9$ by $10$ formation. The site locations are displayed in Figure \ref{sitelocs}.

\begin{figure}
	\begin{center}
		\includegraphics[page=1,width=9cm]{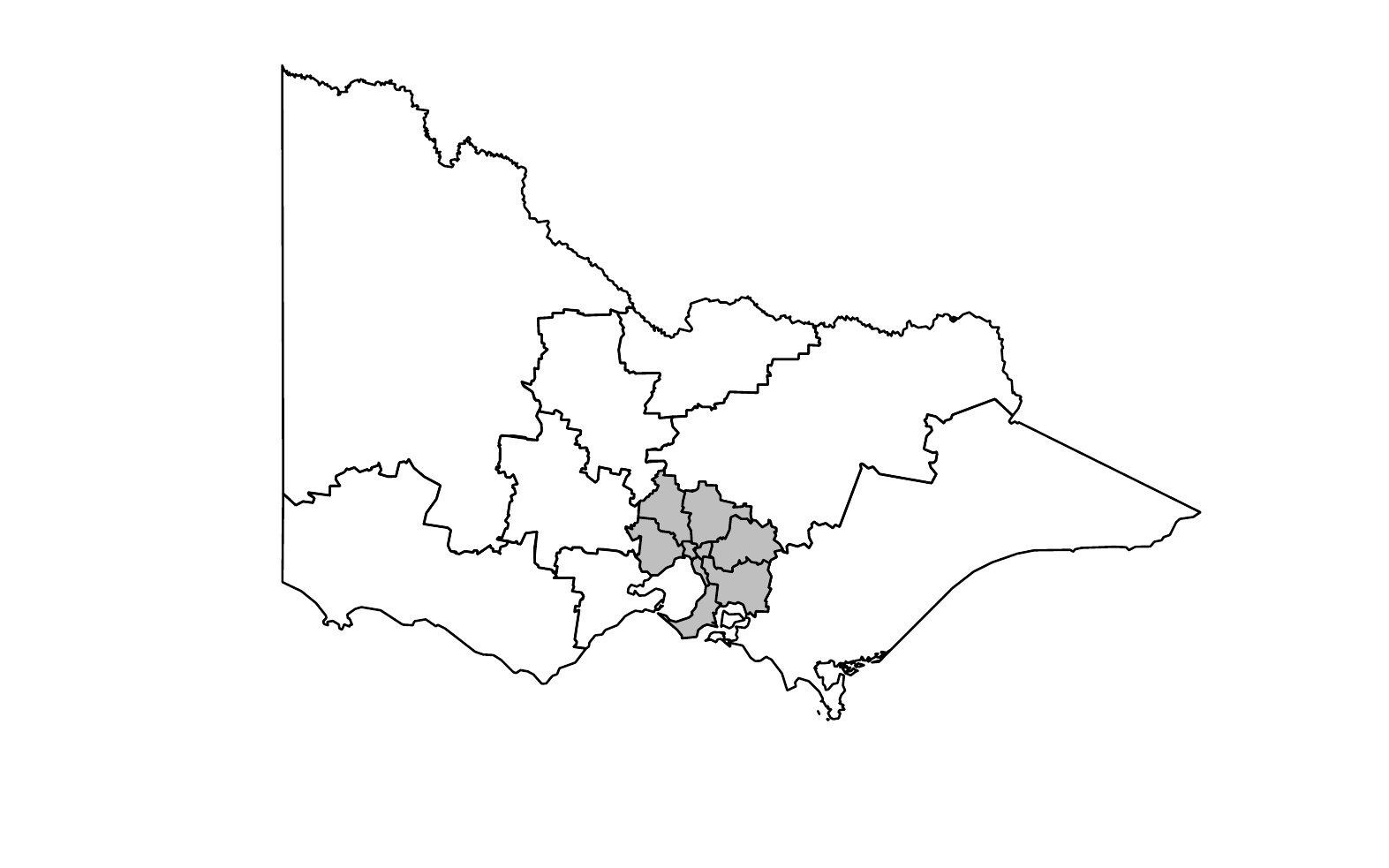}
		\hspace{-2.5cm}
		\includegraphics[page=2,width=9cm]{sitelocs}
	\end{center}
	\vspace{0mm}
	\caption{Site locations. (Left) The Australian state of Victoria with the inner Melbourne region highlighted. (Right) The inner Melbourne region with site locations on a $9$ by $10$ grid.}
	\label{sitelocs}
\end{figure}

The data consists of summer temperature maxima, taken over the extended summer period from August to April inclusive. The first maximum is taken over the August 1961 to April 1962 period, and the last maximum is taken over the August 2010 to April 2011 period. The maxima showed no evidence of temporal non-stationarity.

We additionally know the day of the year on which the temperature maxima occur. This should not be used directly for the hitting scenario, because Melbourne heatwaves often last two or three days. Instead we consider two maxima to belong to the same event if they occur within three days of each other. For each year we therefore derive a hitting scenario $\Pi$, as defined in Section \ref{ssec:inference}. We then use the ST likelihood, which for a single year is given in equation \eqref{eq:STlikelihood}. We use full likelihood inference in preference to $j$-wise likelihood; this is feasible with $d=90$ dimensions due to the quasi-Monte Carlo approximations of the multivariate $t$ distribution function (see Section \ref{sec:sim}). We also employ the powered exponential correlation function $\rho(h) = \exp\{- \|h \| / r)^\eta\}$ with range parameter $r > 0$ and smooth parameter $0 < \eta \leq 2$. 

We first fit the marginal distributions using unconstrained location and scale parameters and shape parameter $\xi = \xi_{0} + \xi_E x_{E} + \xi_N x_{N}$, where $x_{E}$ and $x_{N}$ are (centred) easting and northings in $100$ kilometre units. This gave $\hat{\xi}_{0} = -0.14(0.01)$, $\hat{\xi}_E = 0.02(0.02)$ and $\hat{\xi}_N = 0.09(0.02)$, with location and scale parameters as given in Figure \ref{locsigmaps}. We then use marginal transformations to fit the dependence structure. 

\begin{figure}
	\begin{center}
		\includegraphics[page=1,width=7cm]{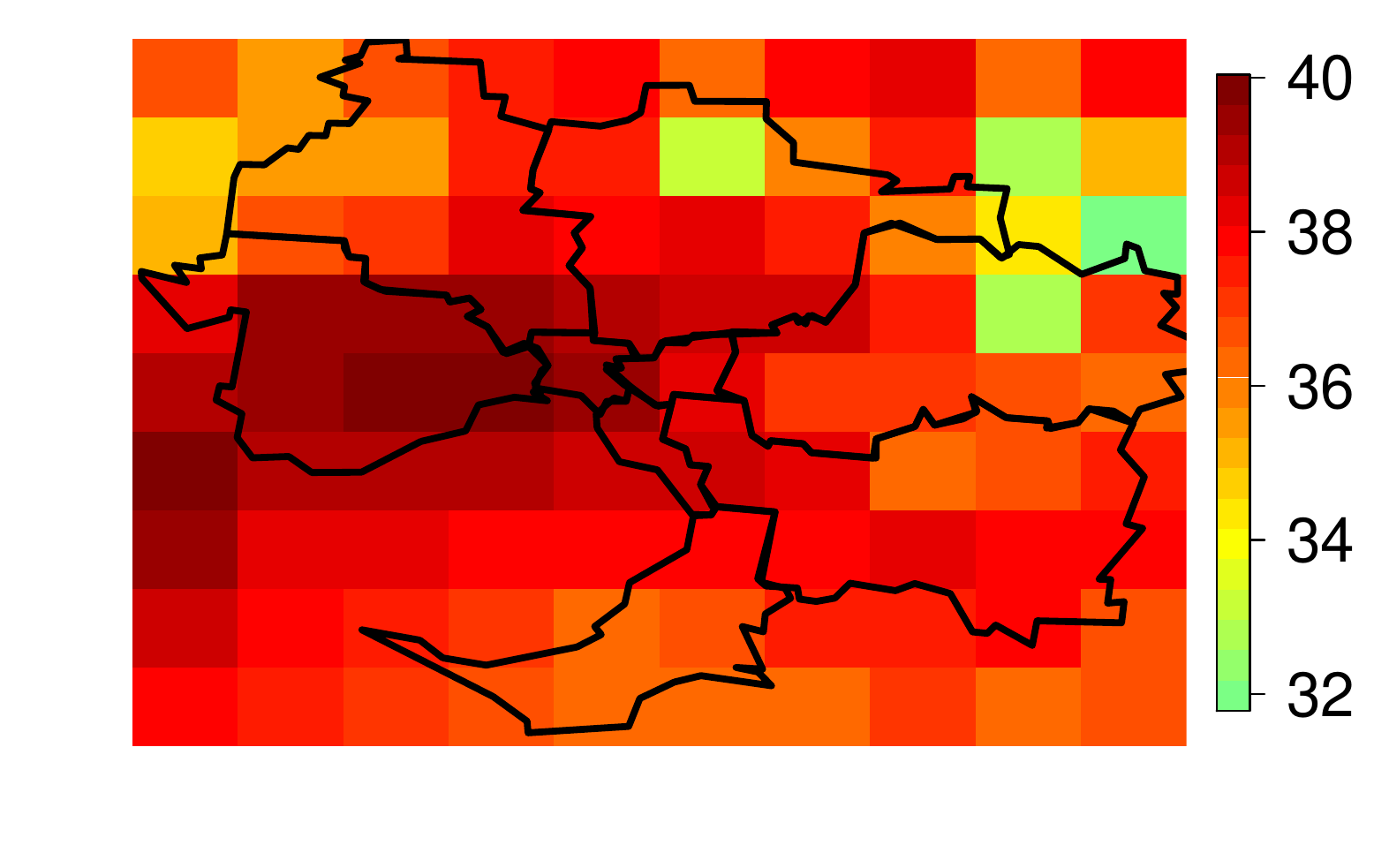}
		\hspace{-0cm}
		\includegraphics[page=2,width=7cm]{locsigmaps}
	\end{center}
	\vspace{-0.5cm}
	\caption{Estimated marginal location (left) and scale (right) parameters.}
	\label{locsigmaps}
\end{figure}

It can be difficult to estimate $r$, $\eta$, and the degrees of freedom parameter $\nu$ simultaneously, so we therefore used a grid search over $\nu=1,3,5$. We additionally model the skewness as $\alpha = \alpha_{0} + \alpha_E x_{E} + \alpha_N x_{N}$. At each value of $\nu$ we optimize over the range $r$, the smooth $\eta$ and the skewness parameters $(\alpha_{0},\alpha_E,\alpha_N)$. With a fixed degree of freedom parameter, the fit of the dependence structure took approximately 2 minutes on a 16 core machine. We then choose the value of $\nu$ producing the largest likelihood. This gave $\hat{\nu} = 5$, $\hat{\eta} = 1.303$ and $\hat{r} = 8.554$, with skewness parameters $\hat{\alpha}_{0} = -0.010$, $\hat{\alpha}_E = -0.281$ and $\hat{\alpha}_N = 0.220$. The largest distance between any two site locations (in 100 kilometre units) is $1.785$, and therefore the smallest correlation is $\exp[-(1.785/\hat{r})^{\hat{\eta}}] \approx 0.88$, indicating a strong degree of spatial dependence. The northern outskirts of Melbourne, particularly to the north-east around Healesville, contains less urban and more elevated terrain, and this may contribute to the selection of the larger value $\hat{\nu} = 5$. The skewness surface is positive to the north-west and negative to the south-east. Fixing the skewness parameters to zero (that is, fitting the extremal $t$ model) gave $\hat{\nu} = 5$, $\hat{\eta} = 1.254$ and $\hat{r} = 8.175$ and a likelihood ratio test highlighted a preference for the fit of the extremal skew-$t$ model.

The fitted dependence structure, in additional to the marginal distributions, can be used for inference on features of interest. Simulations of the process are often required: they can be performed conditionally on the hitting scenario, or conditionally on site observations. Figure \ref{gmmaps} shows two simulations of the process, conditioning on at most two heatwave events causing all annual maxima. 

\begin{figure}
	\begin{center}
		\includegraphics[page=1,width=7cm]{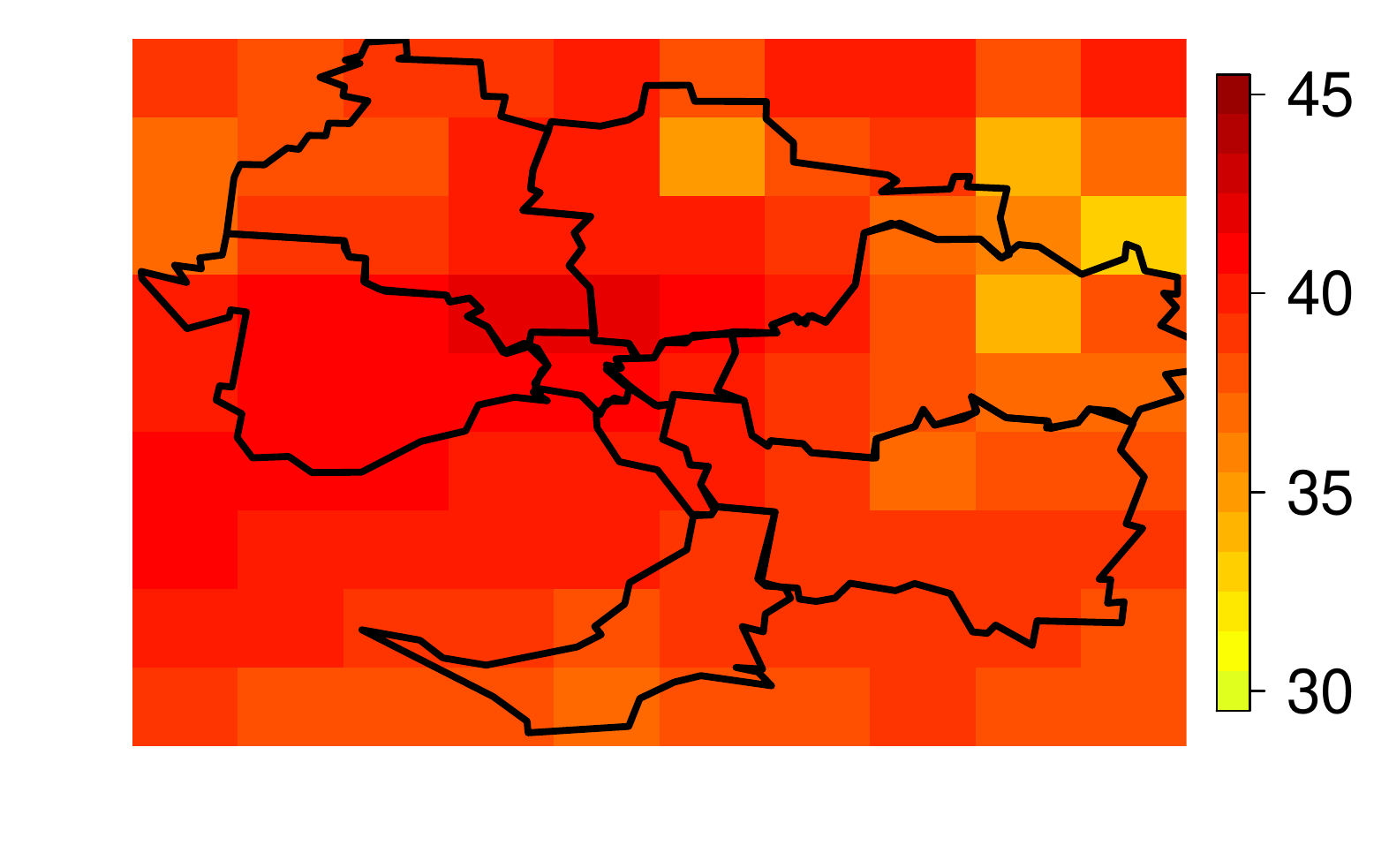}
		\hspace{-0cm}
		\includegraphics[page=2,width=7cm]{gmmaps}
	\end{center}
	\vspace{-0.5cm}
	\caption{Simulations from the fitted max-stable process, conditioning on at most two heatwave events causing all maxima.}
	\label{gmmaps}
\end{figure}

%%%%%%%%%%%%%%%%%%%%%%%%%%%%%%%%%%%%%%%%%%%%
%
% SECTION 5
%
%%%%%%%%%%%%%%%%%%%%%%%%%%%%%%%%%%%%%%%%%%%%
%
\section{Discussion}
\label{sec:discuss}

This article focuses on the general class of extremal skew-$t$ max-stable processes. We first equipped ourselves with the tools required for exact and conditional simulations from the process and derived the necessary results to evaluate the likelihood in any dimension. The known time of occurrence of each maxima has allowed us to use the Stephenson-Tawn likelihood. We proposed strategies to reduce the computational burden associated with the likelihood evaluation using quasi-Monte Carlo approximations. Increasing the dimension at the cost of rougher approximation of the likelihood has proven to be a good strategy: parameter estimation is possible within a reasonable time in dimension up to $d=100$ while maintaining accuracy levels. 

We have proposed combining the Stephenson-Tawn likelihood and composite likelihood methodology. Using this approach we assessed the statistical efficiency of high-order composite likelihood methods and examined their computational cost. Our simulation study outlines a reduction of the root mean squared error for higher-degree composite likelihoods under a fixed degree of approximation and equal numbers of tuples. Our results for high-dimensional data suggest that the $4$-wise composite likelihood is, under most scenarios, a good estimation method for the extremal skew-$t$. We have presented estimation strategies for the flexible extremal skew-$t$  process which are relatively fast, even in the presence of high-dimensional data. We have successfully applied them to a $90$ dimensional temperature dataset recorded in Melbourne, Australia. 

The results presented in this work could potentially be expanded upon by extending the hierarchical matrix decompositions of \citet{genton2018} to multivariate-$t$ cdfs. The selection of an optimal threshold in the definition of the binary weights in the composite likelihood function could also be examined \citep{sang2014, castruccio2016}. Furthermore some of the solutions suggested by \citet{azzalini1999, azzalini2008, azzalini2013} could be implemented to reduce sporadic inaccuracies in the estimation of the skewness.

\section*{Acknowledgements}
This research was undertaken with the assistance of resources and services from the National Computational Infrastructure (NCI), which is supported by the Australian Government. The authors acknowledge Research Technology Services at UNSW Sydney for supporting this project with compute resources. SAS and BB are supported by the Australian Centre of Excellence for Mathematical and Statistical Frontiers (ACEMS; CE140100049) and the Australian Research Council Discovery Project scheme (FT170100079). 

\bibliographystyle{chicago}
\bibliography{biblio_msp}

\begin{thebibliography}{}

\bibitem[\protect\citeauthoryear{Arellano-Valle and Genton}{Arellano-Valle and
  Genton}{2010}]{arellano2010}
Arellano-Valle, R.~B. and M.~G. Genton (2010).
\newblock Multivariate extended skew-{$t$} distributions and related families.
\newblock {\em Metron\/}~{\em 68\/}(3), 201--234.

\bibitem[\protect\citeauthoryear{Azzalini and Arellano-Valle}{Azzalini and
  Arellano-Valle}{2013}]{azzalini2013}
Azzalini, A. and R.~B. Arellano-Valle (2013).
\newblock Maximum penalized likelihood estimation for skew-normal and skew-$t$
  distributions.
\newblock {\em Journal of Statistical Planning and Inference\/}~{\em 143\/}(2),
  419 -- 433.

\bibitem[\protect\citeauthoryear{Azzalini and Capitanio}{Azzalini and
  Capitanio}{1999}]{azzalini1999}
Azzalini, A. and A.~Capitanio (1999).
\newblock Statistical applications of the multivariate skew normal
  distribution.
\newblock {\em J. R. Stat. Soc. Ser. B Stat. Methodol.\/}~{\em 61\/}(3),
  579--602.

\bibitem[\protect\citeauthoryear{Azzalini and Genton}{Azzalini and
  Genton}{2008}]{azzalini2008}
Azzalini, A. and M.~G. Genton (2008).
\newblock Robust likelihood methods based on the skew-$t$ and related
  distributions.
\newblock {\em International Statistical Review\/}~{\em 76\/}(1), 106--129.

\bibitem[\protect\citeauthoryear{Beranger and Padoan}{Beranger and
  Padoan}{2015}]{beranger2015}
Beranger, B. and S.~Padoan (2015).
\newblock Extreme dependence models.
\newblock In {\em Extreme Value Modeling and Risk Analysis}, pp.\  325--352--.
  Chapman and Hall/CRC.

\bibitem[\protect\citeauthoryear{Beranger, Padoan, and Sisson}{Beranger
  et~al.}{2017}]{beranger2017}
Beranger, B., S.~A. Padoan, and S.~A. Sisson (2017).
\newblock Models for extremal dependence derived from skew-symmetric families.
\newblock {\em Scandinavian Journal of Statistics\/}~{\em 44\/}(1), 21--45.

\bibitem[\protect\citeauthoryear{Bienven\"{u}e and Robert}{Bienven\"{u}e and
  Robert}{2017}]{bienvenue2016}
Bienven\"{u}e, A. and C.~Y. Robert (2017).
\newblock Likelihood inference for multivariate extreme value distributions
  whose spectral vectors have known conditional distributions.
\newblock {\em Scandinavian Journal of Statistics\/}~{\em 44\/}(1), 130--149.

\bibitem[\protect\citeauthoryear{Blanchet and Davison}{Blanchet and
  Davison}{2011}]{blanchet2011}
Blanchet, J. and A.~C. Davison (2011).
\newblock Spatial modeling of extreme snow depth.
\newblock {\em Ann. Appl. Stat.\/}~{\em 5\/}(3), 1699--1725.

\bibitem[\protect\citeauthoryear{Brown and Resnick}{Brown and
  Resnick}{1977}]{brown1977}
Brown, B.~M. and S.~I. Resnick (1977).
\newblock Extreme values of independent stochastic processes.
\newblock {\em J. Appl. Probability\/}~{\em 14\/}(4), 732--739.

\bibitem[\protect\citeauthoryear{Buishand, de~Haan, and Zhou}{Buishand
  et~al.}{2008}]{buishand2008}
Buishand, T.~A., L.~de~Haan, and C.~Zhou (2008).
\newblock On spatial extremes: with application to a rainfall problem.
\newblock {\em Ann. Appl. Stat.\/}~{\em 2\/}(2), 624--642.

\bibitem[\protect\citeauthoryear{Castruccio, Huser, and Genton}{Castruccio
  et~al.}{2016}]{castruccio2016}
Castruccio, S., R.~Huser, and M.~G. Genton (2016).
\newblock High-order composite likelihood inference for max-stable
  distributions and processes.
\newblock {\em J. Comput. Graph. Statist.\/}~{\em 25\/}(4), 1212--1229.

\bibitem[\protect\citeauthoryear{Cooley, Cisewski, Erhardt, Jeon, Mannshardt,
  Omolo, and Sun}{Cooley et~al.}{2012}]{cooley2012}
Cooley, D., J.~Cisewski, R.~J. Erhardt, S.~Jeon, E.~Mannshardt, B.~O. Omolo,
  and Y.~Sun (2012).
\newblock A survey of spatial extremes: measuring spatial dependence and
  modeling spatial effects.
\newblock {\em REVSTAT\/}~{\em 10\/}(1), 135--165.

\bibitem[\protect\citeauthoryear{Davison and Gholamrezaee}{Davison and
  Gholamrezaee}{2012}]{davison2012c}
Davison, A.~C. and M.~M. Gholamrezaee (2012).
\newblock Geostatistics of extremes.
\newblock {\em Proceedings of the {R}oyal {S}ociety of {L}ondon {S}eries {A}:
  {M}athematical and {P}hysical {S}ciences\/}~{\em 468}, 581--608.

\bibitem[\protect\citeauthoryear{Davison, Padoan, and Ribatet}{Davison
  et~al.}{2012}]{davison2012b}
Davison, A.~C., S.~A. Padoan, and M.~Ribatet (2012).
\newblock Statistical modeling of spatial extremes.
\newblock {\em Statistical {S}cience\/}~{\em 27}, 161--186.

\bibitem[\protect\citeauthoryear{de~Fondeville and Davison}{de~Fondeville and
  Davison}{2018}]{defondeville2018}
de~Fondeville, R. and A.~C. Davison (2018).
\newblock High-dimensional peaks-over-threshold inference.
\newblock {\em Biometrika\/}~{\em 105\/}(3), 575--592.

\bibitem[\protect\citeauthoryear{de~Haan}{de~Haan}{1984}]{dehaan1984}
de~Haan, L. (1984).
\newblock A spectral representation for max-stable processes.
\newblock {\em Ann. Probab.\/}~{\em 12\/}(4), 1194--1204.

\bibitem[\protect\citeauthoryear{de~Haan and Ferreira}{de~Haan and
  Ferreira}{2006}]{dehaan2006}
de~Haan, L. and A.~Ferreira (2006).
\newblock {\em Extreme value theory}.
\newblock Springer Series in Operations Research and Financial Engineering.
  Springer, New York.
\newblock An introduction.

\bibitem[\protect\citeauthoryear{Dieker and Mikosch}{Dieker and
  Mikosch}{2015}]{dieker2015}
Dieker, A. and T.~Mikosch (2015).
\newblock Exact simulation of {B}rown-{R}esnick random fields at a finite
  number of locations.
\newblock {\em Extremes\/}~{\em 18\/}(2), 301--314.

\bibitem[\protect\citeauthoryear{Dombry, Engelke, and Oesting}{Dombry
  et~al.}{2016}]{dombry2016}
Dombry, C., S.~Engelke, and M.~Oesting (2016).
\newblock Exact simulation of max-stable processes.
\newblock {\em Biometrika\/}~{\em 103\/}(2), 303--317.

\bibitem[\protect\citeauthoryear{Dombry and Eyi-Minko}{Dombry and
  Eyi-Minko}{2013}]{dombry2013a}
Dombry, C. and F.~Eyi-Minko (2013).
\newblock Regular conditional distributions of continuous max-infinitely
  divisible random fields.
\newblock {\em Electron. J. Probab.\/}~{\em 18}, no. 7, 21.

\bibitem[\protect\citeauthoryear{Dombry, {\'E}yi-Minko, and Ribatet}{Dombry
  et~al.}{2013}]{dombry2013b}
Dombry, C., F.~{\'E}yi-Minko, and M.~Ribatet (2013).
\newblock Conditional simulation of max-stable processes.
\newblock {\em Biometrika\/}~{\em 100\/}(1), 111--124.

\bibitem[\protect\citeauthoryear{Fernandez and Steel}{Fernandez and
  Steel}{1999}]{fernandez1999}
Fernandez, C. and M.~Steel (1999, 03).
\newblock {Multivariate Student-t regression models: Pitfalls and inference}.
\newblock {\em Biometrika\/}~{\em 86\/}(1), 153--167.

\bibitem[\protect\citeauthoryear{Genton, Keyes, and Turkiyyah}{Genton
  et~al.}{2018}]{genton2018}
Genton, M.~G., D.~E. Keyes, and G.~Turkiyyah (2018).
\newblock Hierarchical decompositions for the computation of high-dimensional
  multivariate normal probabilities.
\newblock {\em Journal of Computational and Graphical Statistics\/}~{\em
  27\/}(2), 268--277.

\bibitem[\protect\citeauthoryear{Genton, Ma, and Sang}{Genton
  et~al.}{2011}]{genton2011}
Genton, M.~G., Y.~Ma, and H.~Sang (2011).
\newblock On the likelihood function of {G}aussian max-stable processes.
\newblock {\em Biometrika\/}~{\em 98\/}(2), 481--488.

\bibitem[\protect\citeauthoryear{Genz}{Genz}{1992}]{genz1992}
Genz, A. (1992).
\newblock Numerical computation of multivariate normal probabilities.
\newblock {\em Journal of Computational and Graphical Statistics\/}~{\em
  1\/}(2), 141--149.

\bibitem[\protect\citeauthoryear{Genz}{Genz}{1993}]{genz1993}
Genz, A. (1993).
\newblock A comparison of methods for numerical computation of multivariate
  {N}ormal probabilities.
\newblock {\em Computing Science and Statistics\/}~{\em 25}, 400--405.

\bibitem[\protect\citeauthoryear{Genz and Bretz}{Genz and
  Bretz}{2002}]{genz2002}
Genz, A. and F.~Bretz (2002).
\newblock Comparison of methods for the computation of multivariate $t$
  probabilities.
\newblock {\em Journal of Computational and Graphical Statistics\/}~{\em
  11\/}(4), 950--971.

\bibitem[\protect\citeauthoryear{Huser and Davison}{Huser and
  Davison}{2013}]{huser2013}
Huser, R. and A.~C. Davison (2013).
\newblock Composite likelihood estimation for the {B}rown-{R}esnick process.
\newblock {\em Biometrika\/}~{\em 100\/}(2), 511--518.

\bibitem[\protect\citeauthoryear{Huser, Davison, and Genton}{Huser
  et~al.}{2016}]{huser2016}
Huser, R., A.~C. Davison, and M.~G. Genton (2016).
\newblock Likelihood estimators for multivariate extremes.
\newblock {\em Extremes\/}~{\em 19\/}(1), 79--103.

\bibitem[\protect\citeauthoryear{Huser, Dombry, Ribatet, and Genton}{Huser
  et~al.}{2019}]{huser2019}
Huser, R., C.~Dombry, M.~Ribatet, and M.~G. Genton (2019).
\newblock Full likelihood inference for max-stable data.
\newblock {\em Stat\/}~{\em 8\/}(1), e218.
\newblock e218 sta4.218.

\bibitem[\protect\citeauthoryear{Huser and Genton}{Huser and
  Genton}{2016}]{huser2016b}
Huser, R. and M.~G. Genton (2016).
\newblock Non-stationary dependence structures for spatial extremes.
\newblock {\em J. Agric. Biol. Environ. Stat.\/}~{\em 21\/}(3), 470--491.

\bibitem[\protect\citeauthoryear{Jeffrey, Carter, Moodie, and Beswick}{Jeffrey
  et~al.}{2001}]{jeffrey2001}
Jeffrey, S.~J., J.~O. Carter, K.~B. Moodie, and A.~R. Beswick (2001).
\newblock Using spatial interpolation to construct a comprehensive archive of
  {A}ustralian climate data.
\newblock {\em Environmental Modelling \& Software\/}~{\em 16\/}(4), 309 --
  330.

\bibitem[\protect\citeauthoryear{Kabluchko, Schlather, and de~Haan}{Kabluchko
  et~al.}{2009}]{kabluchko2009}
Kabluchko, Z., M.~Schlather, and L.~de~Haan (2009).
\newblock Stationary max-stable fields associated to negative definite
  functions.
\newblock {\em Ann. Probab.\/}~{\em 37\/}(5), 2042--2065.

\bibitem[\protect\citeauthoryear{{Liu}, {Blanchet}, {Dieker}, and
  {Mikosch}}{{Liu} et~al.}{2016}]{liu2016}
{Liu}, Z., J.~H. {Blanchet}, A.~B. {Dieker}, and T.~{Mikosch} (2016).
\newblock On optimal exact simulation of max-stable and related random fields.
\newblock arXiv:1609.06001.

\bibitem[\protect\citeauthoryear{Nikoloulopoulos, Joe, and Li}{Nikoloulopoulos
  et~al.}{2009}]{nikoloulopoulos2009}
Nikoloulopoulos, A.~K., H.~Joe, and H.~Li (2009).
\newblock Extreme value properties of multivariate {$t$} copulas.
\newblock {\em Extremes\/}~{\em 12\/}(2), 129--148.

\bibitem[\protect\citeauthoryear{Oesting, Schlather, and Zhou}{Oesting
  et~al.}{2018}]{oesting2018}
Oesting, M., M.~Schlather, and C.~Zhou (2018).
\newblock Exact and fast simulation of max-stable processes on a compact set
  using the normalized spectral representation.
\newblock {\em Bernoulli\/}~{\em 24\/}(2), 1497--1530.

\bibitem[\protect\citeauthoryear{Opitz}{Opitz}{2013}]{opitz2013}
Opitz, T. (2013).
\newblock Extremal $t$ processes: Elliptical domain of attraction and a
  spectral representation.
\newblock {\em Journal of Multivariate Analysis\/}~{\em 122}, 409 -- 413.

\bibitem[\protect\citeauthoryear{Padoan}{Padoan}{2013}]{padoan2013}
Padoan, S.~A. (2013).
\newblock Extreme dependence models based on event magnitude.
\newblock {\em Journal of Multivariate Analysis\/}~{\em 122\/}(0), 1 -- 19.

\bibitem[\protect\citeauthoryear{Padoan, Ribatet, and Sisson}{Padoan
  et~al.}{2010}]{padoan2010}
Padoan, S.~A., M.~Ribatet, and S.~A. Sisson (2010).
\newblock Likelihood-{B}ased {I}nference for {M}ax-{S}table {P}rocesses.
\newblock {\em Journal of the {A}merican {S}tatistical {A}ssociation\/}~{\em
  105\/}(489), 263--277.

\bibitem[\protect\citeauthoryear{Reich and Shaby}{Reich and
  Shaby}{2012}]{reich2012}
Reich, B.~J. and B.~A. Shaby (2012, 12).
\newblock A hierarchical max-stable spatial model for extreme precipitation.
\newblock {\em Ann. Appl. Stat.\/}~{\em 6\/}(4), 1430--1451.

\bibitem[\protect\citeauthoryear{Ribatet}{Ribatet}{2013}]{ribatet2013}
Ribatet, M. (2013).
\newblock Spatial extremes: Max-stable processes at work.
\newblock {\em J. SFdS\/}~{\em 154\/}(2), 156--177.

\bibitem[\protect\citeauthoryear{Sang and Genton}{Sang and
  Genton}{2014}]{sang2014}
Sang, H. and M.~G. Genton (2014).
\newblock Tapered composite likelihood for spatial max-stable models.
\newblock {\em Spat. Stat.\/}~{\em 8}, 86--103.

\bibitem[\protect\citeauthoryear{Schlather}{Schlather}{2002}]{schlather2002a}
Schlather, M. (2002).
\newblock Models for stationary max-stable random fields.
\newblock {\em Extremes\/}~{\em 5\/}(1), 33--44.

\bibitem[\protect\citeauthoryear{Smith}{Smith}{1985}]{smith1985}
Smith, R.~L. (1985).
\newblock Maximum likelihood estimation in a class of nonregular cases.
\newblock {\em Biometrika\/}~{\em 72\/}(1), 67--90.

\bibitem[\protect\citeauthoryear{Smith}{Smith}{1990}]{smith1990a}
Smith, R.~L. (1990).
\newblock Max-stable processes and spatial extremes.
\newblock University of Surrey 1990 technical report.

\bibitem[\protect\citeauthoryear{Stephenson and Tawn}{Stephenson and
  Tawn}{2005}]{stephenson2005}
Stephenson, A. and J.~Tawn (2005).
\newblock Exploiting occurrence times in likelihood inference for componentwise
  maxima.
\newblock {\em Biometrika\/}~{\em 92\/}(1), 213--227.

\bibitem[\protect\citeauthoryear{Thibaud, Aalto, Cooley, Davison, and
  Heikkinen}{Thibaud et~al.}{2016}]{thibaud2016}
Thibaud, E., J.~Aalto, D.~S. Cooley, A.~C. Davison, and J.~Heikkinen (2016,
  12).
\newblock Bayesian inference for the brown--resnick process, with an
  application to extreme low temperatures.
\newblock {\em Ann. Appl. Stat.\/}~{\em 10\/}(4), 2303--2324.

\bibitem[\protect\citeauthoryear{Thibaud and Opitz}{Thibaud and
  Opitz}{2015}]{thibaud2015}
Thibaud, E. and T.~Opitz (2015).
\newblock Efficient inference and simulation for elliptical {P}areto processes.
\newblock {\em Biometrika\/}~{\em 102\/}(4), 855--870.

\bibitem[\protect\citeauthoryear{Wadsworth}{Wadsworth}{2015}]{wadsworth2015}
Wadsworth, J.~L. (2015).
\newblock On the occurrence times of componentwise maxima and bias in
  likelihood inference for multivariate max-stable distributions.
\newblock {\em Biometrika\/}~{\em 102\/}(3), 705--711.

\bibitem[\protect\citeauthoryear{Wadsworth and Tawn}{Wadsworth and
  Tawn}{2014}]{wadsworth2014}
Wadsworth, J.~L. and J.~A. Tawn (2014).
\newblock Efficient inference for spatial extreme value processes associated to
  log-{G}aussian random functions.
\newblock {\em Biometrika\/}~{\em 101\/}(1), 1--15.

\bibitem[\protect\citeauthoryear{Wang and Stoev}{Wang and
  Stoev}{2011}]{wang2011}
Wang, Y. and S.~A. Stoev (2011).
\newblock Conditional sampling for spectrally discrete max-stable random
  fields.
\newblock {\em Adv. in Appl. Probab.\/}~{\em 43\/}(2), 461--483.

\bibitem[\protect\citeauthoryear{Warne, Sisson, and Drovandi}{Warne
  et~al.}{2019}]{warne+sd19}
Warne, D.~J., S.~A. Sisson, and C.~C. Drovandi (2019).
\newblock Acceleration of expensive computations in {Bayesian} statistics using
  vector operations.
\newblock {\em https://arxiv.org/abs/1902.09046\/}.

\bibitem[\protect\citeauthoryear{{Whitaker}, {Beranger}, and
  {Sisson}}{{Whitaker} et~al.}{2019}]{whitaker2019}
{Whitaker}, T., B.~{Beranger}, and S.~A. {Sisson} (2019).
\newblock {Composite likelihood methods for histogram-valued random variables}.
\newblock {\em arXiv e-prints\/}, arXiv:1908.11548.

\end{thebibliography}

%%%%%%%%%%%%%%%%%%%%%%%%%%%%%%%%%%%%%%%%%%%%%
%
% Section: APPENDIX
%
%%%%%%%%%%%%%%%%%%%%%%%%%%%%%%%%%%%%%%%%%%%%%

\appendix
\section{Technical details}
%

%
% THE NON-CENTRAL EXTENDED SKEW-T
%
\subsection{The non-central extended skew-$t$ distribution \citep{beranger2017}}
\label{app:def_ncextst}

\begin{defi}\label{def:noncen_skew} 

$Y$ is a $d$-dimensional, non-central extended skew-$t$ distributed random vector, denoted by
$Y\sim \mathcal{ST}_d(\mu, \Omega, \alpha, \tau,\kappa,\nu)$, if for  $y\in\real^d$ it has pdf
\begin{equation*}
\psi_d(y;\mu,\Omega,\alpha,\tau,\kappa,\nu)=\frac{\psi_d(y;\mu,\Omega,\nu)}
{\Psi\left(\frac{\tau}{\sqrt{1+Q_{\bar{\Omega}}(\alpha)}};\frac{\kappa}{\sqrt{1+Q_{\bar{\Omega}}(\alpha)}},\nu\right)} 
\Psi \left\{(\alpha^{\top} z+\tau) \sqrt{\frac{\nu + d}{\nu + Q_{\bar{\Omega}^{-1}}(z)}};\kappa,\nu+d\right\},
\end{equation*}
where  $\psi_d(y;\mu,\Omega,\nu)$ is the pdf of a $d$-dimensional $t$-distribution with location $\mu\in\real^d$, $d\times d$ scale matrix
$\Omega$ and $\nu\in\mathbb{R}^+$ degrees of freedom, $\Psi(\cdot;a,\nu)$ denotes a univariate non-central $t$ cdf with non-centrality parameter $a\in\real$ and $\nu$ degrees of freedom, and $Q_{\bar{\Omega}^{-1}}(z)=z^{\top} \bar{\Omega}^{-1} z$, $z=(y - \mu)/\omega$, $\omega=\diag(\Omega)^{1/2}$, 
$\bar{\Omega}=\omega^{-1} \, \Omega\, \omega^{-1}$ and  $Q_{\bar{\Omega}}(\alpha)=\alpha^\top \bar{\Omega}\alpha$. 
The associated cdf is
\begin{equation*}
\Psi_d(y;\mu,\Omega,\alpha,\tau,\kappa,\nu) = 
\frac{\Psi_{d+1} \left\{
\bar{z};
\Omega^*,\kappa^*, \nu
\right\}}
{\Psi\left(\bar{\tau}; \bar{\kappa},\nu\right)},
\end{equation*}
where $\bar{z}=(z^\top,\bar{\tau})^{\top}$, $\Psi_{d+1}$ is a $(d+1)$-dimensional (non-central) $t$ cdf with covariance matrix
and non-centrality parameters
$$
\Omega^*=\left( \begin{array}{cc} 
\bar{\Omega} & - \delta\\
 - \delta^\top & 1 
 \end{array} \right),
\quad \kappa^*=\left( \begin{array}{c} 0 \\ \bar{\kappa} \end{array} \right),
$$
and $\nu$ degrees of freedom,  
and where 
\begin{equation*}
\delta = \left\{ 1+Q_{\bar{\Omega}}(\alpha) \right\}^{-1/2}\,\bar{\Omega}\alpha,\quad
\bar{\kappa} = \left\{ 1+Q_{\bar{\Omega}}(\alpha) \right\}^{-1/2}\, \kappa, \quad
\bar{\tau} = \left\{ 1+Q_{\bar{\Omega}}(\alpha) \right\}^{-1/2}\,\tau. 
\end{equation*}
\end{defi}

%
% Extremal skew t - Exponent function details
%
\subsection{Parameters of the exponent function of the extremal skew-$t$}
\label{app:ext_st_expo}

First of all, we define $m_{j+} \equiv m_+(s_j)$ as follows
\begin{align*}
m_{j+}=\int_0^\infty y_j^\nu\,
\phi(y_j;\alpha^*_{j},\tau^*_{j})\der y_j
=
\frac{2^{(\nu-2)/2} \Gamma\{(\nu+1)/2\}
\Psi(\alpha^*_{j}
\sqrt{\nu+1};-\tau^*_{j},\nu+1) 
}
{\sqrt{\pi}\Phi[\tau\{1+Q_{\bar{\Omega}}(\alpha)\}^{-1/2}]},
\end{align*}
where $\alpha^*_{j}$ and $\tau^*_{j}$ are respectively the marginal shape and extension parameters and $Q_{\bar{\Omega}}(\alpha) = \alpha^\top \bar{\Omega} \alpha$.
This then allows us to obtain the exponent function of the extremal skew-$t$ given in \eqref{eq:V_ext_st} as the sum of $(d-1)$-dimensional non-central extended skew-$t$ distribution with correlation matrix 
$\bar{\Omega}_i^\circ=\omega_{I_iI_i\cdot i}^{-1}\,\bar{\Omega}_{I_iI_i\cdot i}\,\omega_{I_iI_i\cdot i}^{-1}$, where 
$\omega_{I_iI_i\cdot i}=\text{diag}(\bar{\Omega}_{I_iI_i\cdot i})^{1/2}$,
$\bar{\Omega}_{I_iI_i\cdot i}=\bar{\Omega}_{I_iI_i}-\bar{\Omega}_{I_ii}\bar{\Omega}_{iI_i}$,
$I = \{1, \ldots, d \}$, $I_i = I \backslash i$,
$\bar{\Omega}=\omega^{-1} \, \Omega\, \omega^{-1}$,
$\omega=\diag(\Omega)^{1/2}$.
The shape parameter is $\alpha_i^{\circ}=\omega_{I_iI_i\cdot i}\,\alpha_{I_i} \in \real^{d-1}$, 
the extension parameter $\tau_{j}^{\circ}=(\bar{\Omega}_{jI_j}\alpha_{I_j}+\alpha_j)(\nu+1)^{1/2} \in \real$, 
the non-centrality 
$\kappa^{\circ}_j=-\{1+Q_{\bar{\Omega}_{I_jI_j\cdot j}}(\alpha_{I_j})\}^{-1/2}\tau  \in \real$
$Q_{\bar{\Omega}_{I_jI_j\cdot j}}(\alpha_{I_j}) = \alpha_{I_j}^\top \bar{\Omega}_{I_jI_j\cdot j} \alpha_{I_j}$,
and the degrees of freedom $\nu+1$.

%
% Proof of Prop 1 -  p0 extremal skew t
%
\subsection{Proof of Proposition~\ref{prop:p0_ext_st}}
\label{app:exact_ext_st}
\begin{lem}
\label{lem:proof_p0_ext_st}
The finite $d$-dimensional distribution of the random process $(W(s) / W(s_0))_{s \in \cS}$ under the transformed probability measure 
$\widehat{\Pr} = \{ W(s_0)\}_+^{\nu} / m_+(s_0)  \mathrm{d}\Pr $ 
is equal to the distribution of a non-central extended skew-$t$ process with mean $\mu_d$, $d\times d$ scale matrix
$\hat{\Sigma}_d$, slant vector $\hat{\alpha}_d$, extension $\tau_d$, non-centrality $\kappa_d = -\tau$ and $\nu_d = \nu + 1$
degrees of freedom, given by
$$
\mu_d = \Sigma_{d;0}, 
\quad \hat{\Sigma}_d = \frac{\Sigma_d - \Sigma_{d;0} \Sigma_{0;d} }{\nu +1 },
\quad \hat{\alpha}_d = \sqrt{\nu +1} \hat{\omega} \omega_d^{-1} \alpha,
\quad \tau_d = (\alpha_0 + \Sigma_{0;d} \omega_d^{-1} \alpha  ) \sqrt{\nu + 1},
$$
and where $\Sigma_d = (K(x_i, x_j))_{1\leq i,j, \leq d}$,
$\Sigma_{d;0} = \Sigma_{0;d}^\top = (K(x_0, x_i)))\,{1 \leq i \leq d}$,
$\alpha = ( \alpha_1, \ldots, \alpha_d )$,
$\omega_d = \diag(\Sigma_d)^{1/2}$ and $\hat{\omega} = \diag(\hat{\Sigma}_d)^{1/2}$.
\end{lem}
\begin{proof}[Proof of Lemma~\ref{lem:proof_p0_ext_st}]
The proof runs along the same lines as the proof of Lemma~2 in the supplementary material of 
\citet{dombry2016}. We consider finite dimensional distributions only.
Let $d \geq 1$ and $s_1, \ldots, s_d \in \cS$.
We assume that the covariance matrix $\tilde{\Sigma} = (K(s_i, s_j))_{0 \leq i,j \leq d}$ is non
singular so that $(W(s_i))_{0 \leq i \leq d}$ has density 
$$
\tilde{g}(y) = (2 \pi)^{-(d+1)/2} \det (\tilde{\Sigma})^{-1/2} 
\exp\left\{ -\frac{1}{2} y^\top \tilde{\Sigma}^{-1} y \right\}
\frac{\Phi(\tilde{\alpha}^\top \tilde{\omega}^{-1} y + \tau)}{\Phi(\tau/ \sqrt{1 + Q_{\tilde{\bar{\Sigma}}^{-1}}(\tilde{\alpha})}) }, \quad y \in \real^{d+1},
$$
where $\tilde{\alpha} = ( \alpha_0, \alpha_1, \ldots, \alpha_d )$, 
$Q_{\tilde{\bar{\Sigma}}^{-1}}(\tilde{\alpha}) = \tilde{\alpha}^\top \tilde{\bar{\Sigma}} \tilde{\alpha}$, 
$\tilde{\bar{\Sigma}} = \tilde{\omega}^{-1} \tilde{\Sigma} \tilde{\omega}^{-1}$ and $\tilde{\omega} = \diag(\tilde{\Sigma})^{1/2}$.
Setting $z = (y_i /y_0)_{1 \leq i \leq d}$, for all Borel sets $A_1, \ldots, A_d \subset \real$,
\begin{align*}
\widehat{\Pr} \left\{ \frac{W(s_i)}{W(s_0)} \in A_i, i=1, \ldots, d \right\}
&= \int_{\real^{d+1}} \indi \{ y_i/y_0 \in A_i, i=1, \ldots, d \} \frac{(y_0)_+^\nu \tilde{g}(y)}{m_{0+}} 
\der y \\
&= \int_{\real^d} \indi \{ z_i \in A_i, i=1, \ldots, d \} 
\left\{ \int_0^\infty \frac{(y_0)_+^\nu \tilde{g}( (y_0, y_0 z)^\top)}{m_{0+}} y_0^d \der y_0 \right\} \der z.
\end{align*}
We deduce that under $\widehat{\Pr}$ the random vector $(W(s_i)/W(s_0))_{1 \leq i \leq d}$ has 
density
\begin{align*}
g(z) &= \int_0^\infty \frac{y_0^{d+\nu}}{m_{0+}} \tilde{g}((y_0, y_0 z)^\top) \der y_0 \\
&= \frac{(2\pi)^{-(d+1)/2} \det(\tilde{\Sigma})^{-1/2}}
{m_{0+} \Phi(\tau/ \sqrt{1 + Q_{\tilde{\bar{\Sigma}}^{-1}}(\tilde{\alpha})})}
\int_0^\infty y_0^{d+\nu} 
\exp\left\{ -\frac{1}{2} \tilde{z}^\top \tilde{\Sigma}^{-1} \tilde{z} y_0^2 \right\}
\Phi(\tilde{\alpha}^\top \tilde{\omega}^{-1} (y_0, y_0 z)^\top + \tau) \der y_0,
\end{align*}
where $\tilde{z} = (1, z)^\top$. Through the change of variable 
$u = \left( \tilde{z}^\top \tilde{\Sigma}^{-1} \tilde{z} \right)^{1/2} y_0$, we obtain
\begin{align*}
\int_0^\infty y_0^{d+\nu} &
\exp\left\{ -\frac{1}{2} \tilde{z}^\top \tilde{\Sigma}^{-1} \tilde{z} y_0^2 \right\}
\Phi(\tilde{\alpha}^\top \tilde{\omega}^{-1} (y_0, y_0 z) + \tau) \der y_0 \\
&= \sqrt{2\pi} \left( \tilde{z} \tilde{\Sigma}^{-1} \tilde{z}  \right)^{-\frac{d + \nu +1}{2}}
\int_0^\infty u^{d + \nu} \phi(u) 
\Phi \left( \frac{\tilde{\alpha}^\top \tilde{\omega}^{-1} \tilde{z} }
{ \sqrt{ \tilde{z} \tilde{\Sigma}^{-1} \tilde{z} } } u + \tau \right) \der u \\
&= \left( \tilde{z} \tilde{\Sigma}^{-1} \tilde{z}  \right)^{-\frac{d_\nu}{2}}
2^{\frac{\nu -  3}{2}} \pi^{-\frac{d+2}{2}} \Gamma\left( \frac{d_\nu}{2}\right)
\Psi \left( \frac{ \tilde{\alpha}^\top \tilde{\omega}^{-1} \tilde{z} }
{\sqrt{ \tilde{z} \tilde{\Sigma}^{-1} \tilde{z} } } \sqrt{d_\nu}; -\tau, d_\nu \right),
\end{align*}
where $\alpha = (\alpha_1, \ldots, \alpha_d)$, $d_\nu = d + \nu + 1$ and $\Psi(\cdot; \kappa, \nu)$ is the cdf of the non-central $t$ distribution with non-centrality 
parameter $\kappa$ and $\nu$ degrees of freedom.
Thus applying the definition of $m_{0+}$, $\alpha^*_0$ and $\tau^*_0$ from 
\citet{beranger2017} we get
$$
g(z) = \frac{\pi^{-d/2} \det(\tilde{\Sigma})^{-1/2} }
{ \Psi(\alpha^*_0 \sqrt{\nu +1}; -\tau^*_0, \nu +1) }
\left( \tilde{z} \tilde{\Sigma}^{-1} \tilde{z}  \right)^{-\frac{d_\nu}{2}}
\frac{\Gamma\left( \frac{d_\nu}{2}\right)}{\Gamma\left( \frac{\nu + 1}{2}\right)}
\Psi \left( \frac{\tilde{\alpha}^\top \tilde{\omega}^{-1} \tilde{z}}
{\sqrt{ \tilde{z} \tilde{\Sigma}^{-1} \tilde{z} } } \sqrt{d_\nu}; -\tau, d_\nu \right),
$$
and the block decomposition 
$\tilde{\Sigma} = \left(
\begin{array}{cc}
1 & \Sigma_{0;d} \\
\Sigma_{d;0} & \Sigma_{d}
\end{array}
\right)$,
allows us to write
\begin{align*}
g(z) =
\frac{\psi_d(z; \mu_d, \hat{\Sigma}_d, \nu + 1 ) \Psi \left( \frac{\alpha_0 + \alpha^\top \omega_d^{-1} z}
{\sqrt{ \tilde{z} \tilde{\Sigma}^{-1} \tilde{z} } } \sqrt{d_\nu} ; -\tau, d_\nu \right)}
{ \Psi(\alpha^*_0 \sqrt{\nu +1}; -\tau^*_0, \nu +1)},
\end{align*}
where $\mu_d = \Sigma_{d;0}$, $\hat{\Sigma}_d = (\Sigma_d - \Sigma_{d;0} \Sigma_{0;d}) / (\nu + 1)$ and $\omega_d = \diag(\Sigma_d)^{1/2}$. Noting that 
\begin{align*}
\Psi(\alpha^*_0 \sqrt{\nu +1}; -\tau^*_0, \nu +1)
&= \Psi \left( \frac{\tau_d}
{\sqrt{1 + \hat{\alpha}_d^\top \hat{\bar{\Sigma}}_d \hat{\alpha}_d }};
\frac{-\tau}{\sqrt{1 + \hat{\alpha}_d^\top \hat{\bar{\Sigma}}_d \hat{\alpha}_d }};
\nu + 1
\right) \\
 \Psi \left( \frac{\alpha_0 + \alpha^\top \omega_d^{-1} z }
{\sqrt{ \tilde{z} \tilde{\Sigma}^{-1} \tilde{z} } } \sqrt{d_\nu} ; -\tau, d_\nu \right)
&= \Psi \left(  
\frac{ \hat{\alpha}_d z' + \tau_d
}
{\sqrt{\nu + 1 + z'^\top \hat{\bar{\Sigma}}_d^{-1}  z'
}}
\sqrt{d_\nu };
- \tau, d_\nu
\right),
\end{align*}
where $\Psi(\cdot; \kappa, \nu)$ denotes the cdf of the univariate non-central $t$ distribution with non-centrality $\kappa$ and $\nu$ degrees of freedom, $z' = \hat{\omega}^{-1} (z - \Sigma_{d;0})$, $\hat{\alpha}_d = \sqrt{\nu + 1} \hat{\omega} \omega_d^{-1} \alpha$, $\tau_d = (\alpha_0 + \alpha^\top \Sigma_{d;0})
\sqrt{\nu + 1 }$, $\hat{\bar{\Sigma}}_d = \hat{\omega}^{-1} \hat{\Sigma} \hat{\omega}^{-1}$, $\hat{\omega} = \diag(\hat{\Sigma})^{1/2}$ which leads us to the conclusion that 
$g(z)$ is the density of a non-central extended skew-$t$ distribution with parameters $\mu_d = \Sigma_{d;0}$, $\hat{\Sigma}_d$, $\hat{\alpha}_d$, $\tau_d$ and $\kappa_d = -\tau$.

\end{proof}

In order to prove Proposition~\ref{prop:p0_ext_st}, 
let $\mathcal{C}_+ = \mathcal{C} \{\cS, [0, \infty]\}$ denote the space of continuous 
non-negative functions on $\cS$, and $\sigma$ represent the distribution of the 
$\{W(s_i)\}_+^\nu / m_+$, and consider the
set 
$
A = \{ f \in \mathcal{C}_0 : f(s_1) \in A_1, \ldots, f(s_d) \in A_d \}.
$
Then by \citet[Proposition~4.2]{dombry2013a} we have
\begin{align*}
P_{s_0}(A) &= \int_{\mathcal{C}} \indi \{ f/f(s) \in A \} f(s) \sigma(\der f) \\
& = \expect \left[ 
\{W(s_0)\}_+^\nu / m_{0+}
\indi \left\{ \frac{m_{0+} \{W(s_i)\}_+^\nu}{m_{i+} \{ W(s_0)\}_+^\nu} \in A_i; i = 1, \ldots, d\right\} 
\right] \\
&= \widehat{\Pr} \left\{ \frac{m_{0+} \{W(s_i)\}_+^\nu}{m_{i+} \{ W(s_0)\}_+^\nu} 
\in A_i; i = 1, \ldots, d \right\} \\
&= \Pr \left\{ \frac{m_{0+}}{m_{i+}} (T_i)_+^\nu \in A_i; i = 1, \ldots, d \right\},
\end{align*}
where $T = (T_1, \ldots, T_d)$, $T_i = W(s_i) / W(s_0)$ 
which, from Lemma~\ref{lem:proof_p0_ext_st}, is distributed as
$$
T \sim \Psi_d \left(  
\Sigma_{d;0}, 
\frac{\Sigma_d - \Sigma_{d;0} \Sigma_{0;d} }{\nu +1 },
\hat{\alpha}_d,
\tau_d,
- \tau,
\nu + 1
\right).
$$

%
% Proof of Prop 2 - Conditional intensity extremal skew-t
%
\subsection{Proof of Proposition~\ref{prop:cond_intensity_ext_st}}
\label{app:cond_int_ext_st}

The following Lemma is required in order to complete the proof.
\begin{lem}
\label{lem:intensity_function}
Under the assumptions of Proposition~\ref{prop:cond_intensity_ext_st}, the intensity function
of the extremal skew-$t$ is 
$$
\lambda_\bs ( \bv) =  \frac{
2^{(\nu-2)/2} \nu^{-d+1}   
\Gamma \left( \frac{d+\nu}{2} \right) \Psi \left( \tilde{\alpha}_\bs \sqrt{d+\nu}; -\tau_\bs, d+\nu \right) 
\prod_{i=1}^d \left( m_{i+} v_i^{1-\nu} \right)^{1/\nu}
}
{ 
\pi^{d/2} |\bar{\Omega}_\bs|^{1/2} 
 Q_{\bar{\Omega}_\bs}( \bv^\circ)^{(d+\nu)/2} 
\Phi(\tau_\bs \{ 1 + Q_{\bar{\Omega}^{-1}_\bs}(\alpha_\bs) \}^{-1/2} )
},
$$
where
$\bv^\circ = (\bv m_+(\bs))^{1/\nu} \in \real^d$, 
$\tilde{\alpha}_\bs = \alpha_\bs^\top \bv^\circ 
Q_{\bar{\Omega}_\bs}( \bv^\circ)^{-1/2} \in \real$ 
and $\alpha_\bs \in \real^d$.
\end{lem}
\begin{proof}
By definition of the intensity measure \eqref{eq:def_intensity_measure}, 
for all $\bs \in \cS^d$ and Borel set $A \subset \real^d$,
$$
\Lambda_\bs(A) 
= \int_0^\infty \int_{\real^d} \indi \left\{ \zeta \bt^\nu / m_+(\bs) \in A \right\} g_\bs(\bt) \der \bt \, 
\zeta^{-2} \der \zeta, 
$$
where $g_\bs$ is the density of the random vector $W(\bs)$, i.e.~a centred extended skew 
normal random vector with correlation matrix $\bar{\Omega}_\bs$, slant $\alpha_\bs$ and
extension $\tau_\bs$. The change of variable $\bv = (m_+(\bs))^{-1} \zeta \bt^\nu$ leads to 
$\der \bt = \nu^{-d} \zeta^{-d/\nu} \prod_{i=1}^d m_{i+}^{1/\nu} v_i^{(1-\nu)/\nu} \der \bv$
and
\begin{align}
\label{eq:Lambda_fun}
\Lambda_\bs(A) = \nu^{-d} \int_0^\infty \int_A
\prod_{i=1}^d \left( m_{i+} v_i^{(1-\nu)}\right)^{1/\nu}
\frac{
\phi_d \left( \bv^\circ \zeta^{-1/\nu}; \bar{\Omega}_\bs \right)  
\Phi \left( \alpha_\bs^\top \bv^\circ \zeta^{-1/\nu} + \tau_\bs \right) 
}
{
\Phi\left(\tau_\bs \{ 1 + Q_{\bar{\Omega}^{-1}_\bs}(\alpha_\bs) \}^{-1/2} \right)
}
\zeta^{-\frac{d}{\nu}-2} \der \bv \der \zeta,
\end{align}
where $\bv^\circ = (\bv m_+(\bs))^{1/\nu}$.
Now, through the consecutive change of variables $t = \zeta^{-1/\nu}$ and 
$u = t Q_{\bar{\Omega}_\bs}( \bv^\circ)^{1/2}$ 
we  obtain
\begin{align}
\label{eq:intense_fun}
\int_0^\infty \phi_d & \left( \bv^\circ \zeta^{-1/\nu}; \bar{\Omega}_\bs \right)  
\Phi \left( \alpha_\bs^\top \bv^\circ \zeta^{-1/\nu} + \tau_\bs \right) \zeta^{-\frac{d}{\nu}-2} 
\der\zeta \\
\quad & = \nu \int_0^\infty  \phi_d \left( \bv^\circ t ; \bar{\Omega}_\bs \right)  
\Phi \left( \alpha_\bs^\top \bv^\circ t + \tau_\bs \right) t^{d+\nu-1} \der t \nonumber\\
\quad & = (2\pi)^{-d/2} |\bar{\Omega}_\bs|^{-1/2} \nu \int_0^\infty t^{d+\nu-1} 
\exp \left\{ -t^2 \frac{ Q_{\bar{\Omega}_\bs}( \bv^\circ) }{2} \right\} 
\Phi \left( \alpha_\bs^\top \bv^\circ t + \tau_\bs \right) \der t \nonumber\\
\quad & = (2\pi)^{-d/2} |\bar{\Omega}_\bs |^{-1/2} \nu 
Q_{\bar{\Omega}_\bs}( \bv^\circ)^{-(d+\nu)/2}
(2\pi)^{1/2} \int_0^\infty u^{d+\nu-1} \phi(u) 
\Phi \left(\tilde{\alpha}_\bs u + \tau_\bs \right) \der u , \nonumber
\end{align}
where $\tilde{\alpha}_\bs = 
\alpha_\bs^\top \bv^\circ Q_{\bar{\Omega}_\bs}( \bv^\circ)^{-1/2} 
\in \real$.

The remaining integral is linked to the moments of the extended skew-normal distribution.  
\citet[Appendix A.4]{beranger2017} derives the result
$$
\int_0^\infty y^\nu \phi(y) \Phi( \alpha y + \tau) \der y = 2^{(\nu-2)/2} \pi^{-1/2} 
\Gamma \left( \frac{\nu+1}{2} \right) \Psi \left( \alpha \sqrt{\nu+1}; -\tau, \nu+1 \right),
$$
and thus \eqref{eq:intense_fun} is equal to
\begin{equation}
\label{eq:intense_fun_integ}
2^{(\nu-2)/2} \pi^{-d/2} |\bar{\Omega}|^{-1/2} \nu 
\left(\bv^{\circ\top} \bar{\Omega}_\bs^{-1} \bv^\circ\right)^{-(d+1)/2} 
\Gamma \left( \frac{\nu + d}{2} \right) 
\Psi \left(\tilde{\alpha}_\bs  \sqrt{\nu+d}; -\tau_\bs, \nu+d \right).
\end{equation}
Substituting \eqref{eq:intense_fun_integ} into \eqref{eq:Lambda_fun} completes the proof.

\end{proof}

Assume that $(W(\bt), W(\bs)) \sim 
\mathcal{SN}_{m+d}(\bar{\Omega}_{(\bt,\bs)}, \alpha_{(\bt,\bs)}, \tau_{(\bt,\bs)})$,
then according to \citet[Proposition~1]{beranger2017} we have that 
$W(\bs) \sim \mathcal{SN}_d (\bar{\Omega}_\bs, \alpha^*_{\bs}, \tau^*_{\bs})$ with 
\begin{align*}
\begin{array}{ccc}
\alpha^*_\bs = \frac{\alpha_\bs + \bar{\Omega}_\bs^{-1} \bar{\Omega}_{\bs\bt} \alpha_\bt}
{\sqrt{1 + Q_{\tilde{\Omega}^{-1}}(\alpha_\bt)}},
& \tau^*_{\bs} = \frac{\tau_{(\bt,\bs)}}{\sqrt{1 + Q_{\tilde{\Omega}^{-1}}(\alpha_\bt)}},
& \tilde{\Omega} = \bar{\Omega}_\bt - \bar{\Omega}_{\bt \bs}  \bar{\Omega}_\bs^{-1}  
\bar{\Omega}_{\bs \bt}.
\end{array}
\end{align*}
Additionally let $\bu^\circ = (\bu m_+(\bt))^{1/\nu}$, $\bv^\circ = (\bv m_+(\bs))^{1/\nu}$, 
$m_+(\bt)= (m_+(t_1), \ldots, m_+(t_m))$, $m_+(\bs)= (m_+(s_1), \ldots, m_+(s_d))$, 
$\bu \in \real^m$, $\bv \in \real^d$.
Noting that 
$\Phi (\tau_{(\bt,\bs)} (1+Q_{\bar{\Omega}^{-1}_{(\bt,\bs)}}(\alpha_{(\bt,\bs)}))^{-1/2})$ 
is equal to 
$\Phi (\tau^*_{\bs} (1+Q_{\bar{\Omega}^{-1}_{\bs}}(\alpha^*_{\bs}))^{-1/2})$
and applying Lemma \ref{lem:intensity_function} to \eqref{eq:def_cond_int} leads to
\begin{align*}
\lambda_{\bt | \bs, \bv}(\bu) &= 
\pi^{-m/2} \nu^{-m} \frac{|\bar{\Omega}_{(\bt,\bs)}|^{-1/2}}{|\bar{\Omega}_{\bs}|^{-1/2}} 
\left\{ 
\frac{Q_{\bar{\Omega}_{(\bt,\bs)}} (\bu^\circ, \bv^\circ) }
{Q_{\bar{\Omega}_\bs} (\bv^\circ)}
\right\}^{-\frac{\nu+d+m}{2}}
Q_{\bar{\Omega}_\bs} (\bv^\circ)^{-\frac{m}{2}}
\frac{\Gamma(\frac{\nu+d+m}{2})}{\Gamma(\frac{\nu+d}{2})} \\
& \quad \times
\frac{
\Psi\left(\tilde{\alpha}_{(\bt, \bs)} \sqrt{\nu + d + m}; -\tau_{(\bt, \bs)}, \nu + d + m \right)
}
{
\Psi \left(\tilde{\alpha}_{\bs} \sqrt{\nu + d}; -\tau^*_{\bs}, \nu + d \right)
}
\prod_{i=1}^m \left(m_+(t_i) u_i^{1-\nu} \right)^{1/\nu},
\end{align*}
where
$
\tilde{\alpha}_{(\bt, \bs)} = \alpha^\top_{(\bt, \bs)} (\bu^\circ, \bv^\circ) 
Q_{\bar{\Omega}_{(\bt,\bs)}} (\bu^\circ, \bv^\circ)^{-1/2}
$
and
$
\tilde{\alpha}_{\bs} 
= \alpha^{*\top}_\bs \bv^\circ Q_{\bar{\Omega}_\bs} (\bv^\circ)^{-1/2}.
$

Following  \citet{dombry2013b} and \citet{ribatet2013} we can show that
$$
\frac{\lvert\bar{\Omega}_{(\bt,\bs)}\rvert}{\lvert\bar{\Omega}_{\bs}\rvert} 
= \left\{ \frac{\nu + d}{ Q_{\bar{\Omega}} (\bv^\circ) } \right\}^m
\lvert \Omega_{\bt | \bs, \bv} \rvert, \,
%$
%%
%$
\frac{Q_{\bar{\Omega}_{(\bt,\bs)}} (\bu^\circ, \bv^\circ) }
{Q_{\bar{\Omega}_\bs} (\bv^\circ)} 
= 1 + 
\frac{Q_{\Omega_{\bt | \bs, \bv}} (\bu^\circ - \mu_{\bt | \bs, \bv}) }
{\nu + d}, \,
%$
%%
%$
\Omega_{\bt | \bs, \bv} 
= \frac{ Q_{\bar{\Omega}_\bs} (\bv^\circ) }{\nu + d} \tilde{\Omega},
$$
and
$
\mu_{\bt | \bs, \bv} = \bar{\Omega}_{\bt\bs} \bar{\Omega}_\bs^{-1} \bv^\circ.
$
Thus we have
\begin{align*}
\lambda_{\bt | \bs, \bv}(\bu) &= 
\pi^{-m/2} (\nu + d)^{-m/2}  \lvert\Omega_{\bt | \bs, \bv}\rvert^{-1/2} 
\left\{ 1 + 
\frac{Q_{\Omega_{\bt | \bs, \bv}} (\bu^\circ - \mu_{\bt | \bs, \bv}) }
{\nu + d} 
\right\}^{-\frac{\nu + d + m}{2}} \\
& \quad \times \frac{\Gamma(\frac{\nu+d+m}{2})}{\Gamma(\frac{\nu+d}{2})} 
\frac{\Psi \left(\tilde{\alpha}_{(\bt, \bs)} \sqrt{\nu + d + m}; -\tau_{(\bt, \bs)}, \nu + d + m \right)}
{\Psi \left(\tilde{\alpha}_{\bs} \sqrt{\nu + d}; -\tau^*_{\bs}, \nu + d \right)}
\nu^{-m} \prod_{i=1}^m \left(m_+(t_i) u_i^{1-\nu} \right)^{1/\nu},
\end{align*}
and we recognise the $m$-dimensional Student-$t$ density with mean $\mu_{\bt | \bs, \bv}$, 
dispersion matrix $\Omega_{\bt | \bs, \bv}$ and degree of freedom $\nu + d$.

Finally, by considering 
$
\alpha_{\bt | \bs, \bv} = \tilde{\omega}^{-1} \alpha_\bt 
$,
$
\tilde{\omega} = \diag(\tilde{\Omega})^{1/2},
$
$
\tau_{\bt | \bs, \bv}  = 
\left( \alpha_\bs + \bar{\Omega}_\bs^{-1} \bar{\Omega}_{\bs \bt} \alpha_\bt \right)^\top 
\bv^\circ (d+\nu)^{1/2} Q_{\bar{\Omega}} (\bv^\circ)^{-1/2}
$
and
$
\kappa_{\bt | \bs, \bv} = -\tau_{(\bt, \bs)}
$
then it is easy to show that 
\begin{align*}
\Psi \left(\tilde{\alpha}_{\bs} \sqrt{\nu + d}; -\tau^*_{\bs}, \nu + d \right)
&= \Psi \left( \frac{\tau_{\bt | \bs, \bv}}{\sqrt{1+Q_{\bar{\Omega}^{-1}_{\bt | \bs, \bv}}(\alpha_{\bt | \bs, \bv}) }} ; 
-\frac{\tau_{(\bt,\bs)}}{\sqrt{1+Q_{\bar{\Omega}^{-1}_{\bt | \bs, \bv}}(\alpha_{\bt | \bs, \bv}) }}, 
\nu_{\bt | \bs, \bv} \right) \\
\Psi \left(\tilde{\alpha}_{(\bt, \bs)} \sqrt{\nu + d + m}; -\tau_{(\bt, \bs)}, \right. & \left. \nu + d + m \right) \\
&= \Psi \left( \left(\alpha_{\bt | \bs, \bv}^\top z + \tau_{\bt | \bs, \bv} \right) 
\sqrt{\frac{\nu_{\bt | \bs, \bv} + m}{\nu_{\bt | \bs, \bv} + Q_{\bar{\Omega}_{\bt | \bs, \bv}}(z)} }; 
-\tau_{(\bt,\bs)}, 
\nu_{\bt | \bs, \bv}+m \right) ,
\end{align*}
where $z = \bu^\circ - \mu_{\bt | \bs, \bv}$, completes the proof.
Note that $\bar{\Omega}_{\bt | \bs, \bv}$ reduces to $\tilde{\omega}^{-1} \tilde{\Omega} \tilde{\omega}^{-1}$.
%
%
% Partial derivatives of V function Lemma 1
%
\subsection{Partial derivatives of the \texorpdfstring{$V$}{Lg} function of the extremal skew-\texorpdfstring{$t$}{Lg}
(Lemma~\ref{prop:vd_ext_st}) }
\label{app:prop_vd_ext_st}
Consider the conditional intensity function of the extremal skew-$t$ model given in 
Proposition~\ref{prop:cond_intensity_ext_st} with $\bs = (s_1, \ldots, s_m) \equiv \bs_{1:m}$, 
$\bt = (s_{m+1}, \ldots, s_d) \equiv \bs_{m+1:d}$, $\bv = \bz_{1:m}$ and $\bu = \bz_{m+1:d}$. 
In the following the matrix notation $\Sigma_{a:b}= \Sigma_{\bs_{a:b}}, \Sigma_{a:b;c:d}= \Sigma_{\bs_{a:b} \bs_{c:d}}$ will be used.
Integration w.r.t. $\bz_{m+1:d}$ gives
\begin{align}
\label{eq:Lem1_cond}
\Psi_{d-m} \left(\bz_{m+1:d}^\circ; \mu_c, \Omega_c, \alpha_c, \tau_c, \kappa_c, \nu_c \right)
\end{align}
where the index $c$ represents $\bs_{m+1:d}| \bs_{1:m}, \bz_{1:m}$ such that the parameters
are
\begin{align*}
&\mu_c 
= \bar{\Omega}_{m+1:d ; 1:m} \bar{\Omega}_{1:m}^{-1} 
\bz_{1:m}^\circ, \quad
\Omega_c 
= \frac{Q_{\bar{\Omega}_{1:m}} (\bz_{1:m}^\circ)}{\nu_c}  
\tilde{\Omega}_c, \\
&\tilde{\Omega}_c 
= 
\bar{\Omega}_{m+1:d} - \bar{\Omega}_{m+1:d; 1:m} 
\bar{\Omega}_{1:m}^{-1} \bar{\Omega}_{1:m ; m+1:d}, \\
&\tau_c 
 = (\alpha_{1:m} + \bar{\Omega}_{1:m}^{-1} \bar{\Omega}_{1:m ; m+1:d} 
\alpha_{m+1:d} )^\top 
\bz_{1:m}^\circ (\nu+m)^{1/2} 
Q_{\bar{\Omega}_{1:m}} (\bz_{1:m}^\circ)^{-1/2}, \\
&\alpha_c 
= \omega_c \alpha_{m+1:d}, \quad 
\tilde{\omega}_c = \diag(\tilde{\Omega}_c)^{1/2}, \quad
\kappa_c 
= -\tau, \quad
\nu_c 
= \nu + m, 
\end{align*}
with
$\bz_{1:m}^\circ = (\bz_{1:m} m_+(\bs_{1:m}))^{1/\nu}$ and 
$\bz_{m+1:d}^\circ = (\bz_{m+1:d} m_+(\bs_{m+1:d}))^{1/\nu}$.
According to Lemma~\ref{lem:intensity_function}, the $m$-dimensional marginal density is
\begin{align}
\label{eq:Lem1_marginal}
\frac{
2^{(\nu-2)/2} \nu^{-m+1}   
\Gamma \left( \frac{m+\nu}{2} \right) 
\Psi \left( \tilde{\alpha}_{1:m} \sqrt{m+\nu}; - \tau_{1:m}^*, m+\nu \right) 
\prod_{i=1}^m \left( m_{i+} z_i^{1-\nu} \right)^{1/\nu}
}
{ 
\pi^{m/2} |\bar{\Omega}_{1:m}|^{1/2} 
Q_{\bar{\Omega}_{1:m}} (\bz_{1:m}^\circ )^{(m+\nu)/2} 
\Phi(\tau_{1:m}^* \{ 1 + Q_{\bar{\Omega}^{-1}_{1:m}}(\alpha^*_{1:m}) \}^{-1/2} )
},
\end{align}
where
$\tilde{\alpha}_{1:m} = \alpha_{1:m}^{*\top} \bz_{1:m}^\circ 
Q_{\bar{\Omega}_{1:m}} (\bz_{1:m}^\circ )^{-1/2} \in \real$,
and $m$-dimensional marginal parameters 
\begin{align*}
\begin{array}{cc}
\alpha_{1:m}^* 
= \frac{\alpha_{1:m} + \bar{\Omega}_{1:m}^{-1} 
\bar{\Omega}_{1:m; m+1:d} \alpha_{m+1:d} }
{\sqrt{1+ Q_{\bar{\Omega}^{-1}_c} (\alpha_{m+1:d})}},
& \tau_{1:m}^*  = 
\frac{\tau}
{\sqrt{1+ Q_{\bar{\Omega}^{-1}_c}(\alpha_{m+1:d})}}.
\end{array}
\end{align*} 
Combining \eqref{eq:Lem1_cond} and \eqref{eq:Lem1_marginal} completes the proof.

Setting $\tau_\bs = 0$ corresponds to an extremal skew-$t$ model constructed from a skew-normal random field rather than an extended skew-normal field. Then 
\begin{equation*}
\label{eq:mi_plus}
m_{j+} = 2^{\nu/2} \pi^{-1/2} \Gamma\left( \frac{\nu + 1}{2}\right) 
\Psi(\alpha^*_j \sqrt{\nu +1}; \nu +1),
\end{equation*}
with
\begin{equation*} 
\alpha^*_j =  
\frac{\alpha_j + \bar{\Omega}^{-1}_{jj} \bar{\Omega}_{j I_j} \alpha_{I_j} }  
{\sqrt{1 + \alpha_{I_j}^\top \left( \bar{\Omega}_{I_j I_j} - \bar{\Omega}_{I_j j} \bar{\Omega}^{-1}_{jj} \bar{\Omega}_{j I_j} \right) \alpha_{I_j}} },
\end{equation*}
and the associated partial derivatives of the $V$ function are equal to
$$
\Psi_{d-m} \left(\bz_{m+1:d}^\circ; \mu_c, \Omega_c, \alpha_c, \tau_c, 0, \nu_c \right)
\frac{
2^{\nu/2}   
\Gamma \left( \frac{m+\nu}{2} \right) 
\Psi \left( \tilde{\alpha}_{1:m} \sqrt{m+\nu}; m+\nu \right) 
\prod_{i=1}^m \left( m_{i+} z_i^{1-\nu} \right)^{1/\nu}
}
{ 
\pi^{m/2}  \nu^{m-1} |\bar{\Omega}_{1:m}|^{1/2} 
Q_{\bar{\Omega}_{1:m}} (\bz_{1:m}^\circ )^{(m+\nu)/2} 
},
$$
with parameters defined as in Lemma~\ref{prop:vd_ext_st}.

Setting $\alpha_\bs = \alpha_{1:d} = 0 $ and $\tau_\bs = 0$ leads to the extremal-$t$ model for
which 
$$
m_{i+} = 2^{(\nu-2)/2} \pi^{-1/2} \Gamma\left( \frac{\nu + 1}{2}\right) \equiv m_+,
$$
and the partial derivatives of the $V$ function are
$$
\Psi_{d-m} \left(\bz_{m+1:d}^{1/\nu}; \mu'_c, \Omega'_c, \nu_c \right)
\frac{
2^{(\nu-2)/2} \nu^{-m+1}  \Gamma \left( \frac{m+\nu}{2} \right) 
\prod_{i=1}^m z_i^{(1-\nu)/\nu}
}
{ 
\pi^{m/2} |\bar{\Omega}_{1:m}|^{1/2} 
Q_{\bar{\Omega}_{1:m}} (\bz_{1:m}^{1/ \nu} )^{(m+\nu)/2} 
m_+
},
$$
where
$\mu'_c = \bar{\Omega}_{m+1:d ; 1:m} \bar{\Omega}_{1:m}^{-1} \bz_{1:m}^{1/ \nu}$
and
$\Omega'_c 
= \frac{Q_{\bar{\Omega}_{1:m}} (\bz_{1:m}^{1/\nu})}{\nu_c}  
\bar{\Omega}_c.$
\clearpage
\section{Simulation Tables}
\label{app:sim}

\afterpage{
\clearpage
\begin{landscape}
\begin{table}
\centering
\footnotesize
$
\begin{array}{ccc c cccccccc c cccccccccccccc}
\multicolumn{4}{c}{} & \multicolumn{8}{c}{\textrm{extremal-$t$}} & & \multicolumn{14}{c}{\textrm{extremal skew-$t$}} \\
\cline{5-12} \cline{14-27} 
\multicolumn{4}{c}{} & \multicolumn{2}{c}{r=1.5} & & \multicolumn{2}{c}{r=3.0} & & \multicolumn{2}{c}{r=4.5} & & \multicolumn{4}{c}{r=1.5} & & \multicolumn{4}{c}{r=3.0} & & \multicolumn{4}{c}{r=4.5} \\
\cline{5-6} \cline{8-9} \cline{11-12} \cline{14-17} \cline{19-22} \cline{24-27}
& & \mbox{Type} &  & \hat{\eta}_j & \hat{r}_j & & \hat{\eta}_j & \hat{r}_j & & \hat{\eta}_j & \hat{r}_j & & \hat{\eta}_j & \hat{r}_j & \hat{\beta}_{1j} & \hat{\beta}_{2j} & & \hat{\eta}_j & \hat{r}_j & \hat{\beta}_{1j} & \hat{\beta}_{2j} & & \hat{\eta}_j & \hat{r}_j & \hat{\beta}_{1j} & \hat{\beta}_{2j} \\
\hline
d=20 & \eta=1.00 & \mbox{I} & & 0.005 & 0.014 &  & 0.002 & 0.006 &  & 0.001 & 0.007 & & 0.002 & 0.038 & 0.211 & 0.181 &  & 0.013 & 0.031 & 0.111 & 0.094 &  & 0.016 & 0.118 & 0.017 & 0.015 \\ 
& & \mbox{II} & & 0.008 & 0.021 &  & 0.007 & 0.050 &  & 0.004 & 0.072 & & 0.017 & 0.116 & 0.068 & 0.048 &  & 0.018 & 0.173 & 0.096 & 0.128 &  & 0.017 & 0.117 & 0.149 & 0.213 \\
& \eta=1.50 & \mbox{I} & & 0.001 & 0.004 &  & 0.003 & 0.015 &  & 0.002 & 0.003 &  & 0.024 & 0.030 & 0.434 & 0.308 &  & 0.003 & 0.023 & 0.026 & 0.001 &  & 0.000 & 0.170 & 0.069 & 0.053 \\ 
& & \mbox{II} & & 0.005 & 0.015 &  & 0.004 & 0.021 &  & 0.002 & 0.013 &  & 0.010 & 0.051 & 0.072 & 0.059 &  & 0.011 & 0.045 & 0.055 & 0.046 &  & 0.004 & 0.123 & 0.068 & 0.048 \\
& \eta=1.95 & \mbox{I} & & 0.003 & 0.002 &  & 0.000 & 0.003 &  & 0.000 & 0.009 &  & 0.002 & 0.003 & 0.282 & 0.217 &  & 0.000 & 0.047 & 0.099 & 0.078 &  & 0.001 & 0.053 & 0.041 & 0.034 \\
& & \mbox{II} & & 0.000 & 0.006 &  & 0.001 & 0.004 &  & 0.001 & 0.002 &  & 0.001 & 0.022 & 0.013 & 0.011 &  & 0.003 & 0.019 & 0.104 & 0.143 &  & 0.002 & 0.057 & 0.061 & 0.077 \\
\hline
d=50 & \eta=1.00 & \mbox{I} & & 0.001 & 0.003 &  & 0.003 & 0.018 &  & 0.001 & 0.005 &  & 0.004 & 0.013 & 0.013 & 0.005 &  & 0.003 & 0.004 & 0.017 & 0.002 &  & 0.001 & 0.043 & 0.018 & 0.010 \\ 
& & \mbox{II} & & 0.002 & 0.009 &  & 0.003 & 0.027 &  & 0.003 & 0.024 &  & 0.003 & 0.033 & 0.012 & 0.012 &  & 0.005 & 0.010 & 0.006 & 0.009 &  & 0.004 & 0.007 & 0.014 & 0.013 \\
& \eta=1.50 & \mbox{I} & & 0.001 & 0.000 &  & 0.000 & 0.001 &  & 0.001 & 0.001 &  & 0.004 & 0.010 & 0.011 & 0.015 &  & 0.001 & 0.045 & 0.002 & 0.000 &  & 0.002 & 0.053 & 0.008 & 0.003 \\
& & \mbox{II} & & 0.001 & 0.002 &  & 0.000 & 0.004 &  & 0.000 & 0.009 &  & 0.006 & 0.019 & 0.019 & 0.021 &  & 0.001 & 0.015 & 0.003 & 0.015 &  & 0.002 & 0.012 & 0.014 & 0.031 \\
& \eta=1.95 & \mbox{I} & & 0.000 & 0.000 &  & 0.000 & 0.004 &  & 0.000 & 0.003 &  & 0.002 & 0.009 & 0.016 & 0.016 &  & 0.000 & 0.016 & 0.002 & 0.013 &  & 0.000 & 0.009 & 0.004 & 0.012 \\
& & \mbox{II} & & 0.000 & 0.001 &  & 0.000 & 0.000 &  & 0.000 & 0.005 &  & 0.002 & 0.008 & 0.073 & 0.036 &  & 0.001 & 0.017 & 0.015 & 0.028 &  & 0.000 & 0.015 & 0.009 & 0.036 \\
\hline
d=100 & \eta=1.00 & \mbox{I} & & 0.003 & 0.016 &  & 0.002 & 0.037 &  & 0.001 & 0.053 &  & 0.005 & 0.039 & 0.008 & 0.012 &  & 0.002 & 0.009 & 0.009 & 0.003 &  & 0.004 & 0.068 & 0.006 & 0.007 \\
& & \mbox{II} & & 0.007 & 0.036 &  & 0.003 & 0.080 &  & 0.004 & 0.091 &  & 0.006 & 0.036 & 0.005 & 0.001 &  & 0.000 & 0.052 & 0.004 & 0.003 &  & 0.002 & 0.069 & 0.008 & 0.014 \\  
& \eta=1.50 & \mbox{I} & & 0.000 & 0.004 &  & 0.001 & 0.007 &  & 0.000 & 0.009 &  & 0.003 & 0.013 & 0.001 & 0.001 &  & 0.000 & 0.010 & 0.004 & 0.002 &  & 0.001 & 0.090 & 0.003 & 0.052 \\
& & \mbox{II} & & 0.001 & 0.011 &  & 0.000 & 0.014 &  & 0.001 & 0.001 &  & 0.003 & 0.015 & 0.006 & 0.007 &  & 0.001 & 0.004 & 0.009 & 0.017 &  & 0.001 & 0.026 & 0.030 & 0.062 \\
& \eta=1.95 & \mbox{I} & & 0.000 & 0.004 &  & 0.000 & 0.001 &  & 0.000 & 0.000 &  & 0.000 & 0.003 & 0.037 & 0.017 &  & 0.000 & 0.003 & 0.007 & 0.004 &  & 0.001 & 0.029 & 0.079 & 0.123 \\
& & \mbox{II} & & 0.000 & 0.006 &  & 0.000 & 0.004 &  & 0.000 & 0.005 &  & 0.000 & 0.007 & 0.002 & 0.001 &  & 0.001 & 0.002 & 0.005 & 0.013 &  & 0.001 & 0.009 & 0.153 & 0.223 \\
\hline

\end{array}
$
\caption{\small Absolute biases $|\bar{\hat{\theta}}_j - \theta_j|$ for $\hat{\theta}_j = (\hat{\eta}_j, \hat{r}_j)$ and $\hat{\theta}_j = (\hat{\eta}_j, \hat{r}_j, \hat{\beta}_{1j}, \hat{\beta}_{2j})$ the parameter vectors of the extremal-$t$ and extremal skew-$t$ models, using the full likelihood Type I and Type II approximations given in Table~\ref{tab:approx} when $d=20, 50$ and $100$ sites are considered. Calculations are based on 500 replicate maximisations.}
\label{tab:Ext_full_BIAS}
\end{table}
\end{landscape}
}

\afterpage{
\clearpage
\begin{landscape}
\begin{table}
\centering
\footnotesize
$
\begin{array}{ccc c cccccccc c cccccccccccccc}
\multicolumn{4}{c}{} & \multicolumn{8}{c}{\textrm{extremal-$t$}} & & \multicolumn{14}{c}{\textrm{extremal skew-$t$}} \\
\cline{5-12} \cline{14-27} 
\multicolumn{4}{c}{} & \multicolumn{2}{c}{r=1.5} & & \multicolumn{2}{c}{r=3.0} & & \multicolumn{2}{c}{r=4.5} & & \multicolumn{4}{c}{r=1.5} & & \multicolumn{4}{c}{r=3.0} & & \multicolumn{4}{c}{r=4.5} \\
\cline{5-6} \cline{8-9} \cline{11-12} \cline{14-17} \cline{19-22} \cline{24-27}
& & j &  & \hat{\eta}_j & \hat{r}_j & & \hat{\eta}_j & \hat{r}_j & & \hat{\eta}_j & \hat{r}_j & & \hat{\eta}_j & \hat{r}_j & \hat{\beta}_{1j} & \hat{\beta}_{2j} & & \hat{\eta}_j & \hat{r}_j & \hat{\beta}_{1j} & \hat{\beta}_{2j} & & \hat{\eta}_j & \hat{r}_j & \hat{\beta}_{1j} & \hat{\beta}_{2j} \\
\hline
 & \eta=1.00 & j=2 & & 0.096 & 0.190 &  & 0.099 & 0.439 &  & 0.087 & 0.858 & & 0.681 & 5.093 & 4.502 & 4.685 &  & 0.707 & 2.271 & 4.392 & 4.407 &  & 0.624 & 4.569 & 4.170 & 4.386  \\
& & j=3 & & 0.102 & 0.196 &  & 0.083 & 0.503 &  & 0.078 & 0.785 & & 0.419 & 3.042 & 1.500 & 1.908 &  & 0.315 & 3.124 & 1.536 & 1.847 &  & 0.250 & 2.311 & 1.484 & 2.054 \\
& & j=4 & & 0.127 & 0.207 &  & 0.089 & 0.462 &  & 0.077 & 0.769 & & 0.196 & 0.347 & 1.255 & 1.466 &  & 0.130 & 0.915 & 0.917 & 1.464 &  & 0.127 & 3.704 & 1.098 & 1.581 \\
& & j=5 & & 0.125 & 0.222 &  & 0.096 & 0.463 &  & 0.087 & 0.818 & & 0.298 & 2.810 & 1.629 & 1.985 &  & 0.163 & 2.126 & 1.597 & 2.088 &  & 0.145 & 1.343 & 1.548 & 2.012 \\
& & j=10 & & 0.077 & 0.152 &  & 0.063 & 0.312 &  & 0.057 & 0.514 & & 0.250 & 0.313 & 1.456 & 1.615 &  & 0.154 & 0.591 & 1.646 & 1.087 &  & 0.156 & 0.999 & 1.011 & 1.063 \\
& & j=d & & 0.061 & 0.124 &  & 0.048 & 0.246 &  & 0.044 & 0.378 & & 0.106 & 0.223 & 0.380 & 0.394 &  & 0.117 & 0.478 & 0.633 & 0.893 &  & 0.101 & 0.670 & 0.831 & 0.883 \\
\hline
& \eta=1.50 & j=2 & & 0.106 & 0.146 &  & 0.077 & 0.295 &  & 0.076 & 0.515 & & 0.458 & 1.971 & 4.557 & 4.851 &  & 0.379 & 3.549 & 4.248 & 5.037 &  & 0.352 & 5.570 & 4.271 & 4.820 \\
& & j=3 & & 0.110 & 0.136 &  & 0.071 & 0.286 &  & 0.058 & 0.469  & & 0.223 & 1.896 & 1.491 & 1.903 &  & 0.155 & 2.382 & 1.547 & 2.026 &  & 0.128 & 2.208 & 1.371 & 1.938 \\
& & j=4 & & 0.115 & 0.127 &  & 0.079 & 0.277 &  & 0.060 & 0.432 & & 0.134 & 0.170 & 0.902 & 1.360 &  & 0.118 & 0.640 & 1.010 & 1.547 &  & 0.112 & 1.374 & 1.100 & 1.710 \\
& & j=5 & & 0.108 & 0.132 &  & 0.076 & 0.279 &  & 0.064 & 0.465 & & 0.160 & 0.201 & 1.494 & 1.741 &  & 0.113 & 2.431 & 1.544 & 1.893 &  & 0.093 & 0.666 & 1.505 & 1.921 \\
& & j=10 & & 0.071 & 0.097 &  & 0.041 & 0.176 &  & 0.040 & 0.290 & & 0.136 & 0.149 & 1.893 & 1.590 &  & 0.081 & 0.291 & 0.689 & 0.931 &  & 0.069 & 0.556 & 0.977 & 1.017 \\
& & j=d & & 0.045 & 0.074 &  & 0.032 & 0.156 &  & 0.027 & 0.226 & & 0.096 & 0.142 & 0.363 & 0.359 &  & 0.050 & 0.236 & 0.324 & 0.323 &  & 0.039 & 0.345 & 0.413 & 0.401 \\ 
\hline
& \eta=1.95 & j=2 & & 0.042 & 0.111 &  & 0.040 & 0.176 &  & 0.034 & 0.307 & & 0.225 & 2.256 & 4.603 & 5.215 &  & 0.107 & 2.644 & 4.409 & 5.058 &  & 0.087 & 3.081 & 4.372 & 4.525 \\
& & j=3 & & 0.049 & 0.088 &  & 0.034 & 0.185 &  & 0.023 & 0.285 & & 0.060 & 2.544 & 1.975 & 2.346 &  & 0.076 & 1.526 & 1.732 & 2.484 &  & 0.073 & 0.927 & 1.658 & 2.510 \\
& & j=4 & & 0.050 & 0.091 &  & 0.018 & 0.135 &  & 0.013 & 0.216 & & 0.044 & 0.093 & 1.298 & 1.330 &  & 0.054 & 0.289 & 1.125 & 1.563 &  & 0.055 & 0.646 & 1.831 & 2.133 \\
& & j=5 & & 0.060 & 0.090 &  & 0.017 & 0.134 &  & 0.011 & 0.207 & & 0.039 & 0.112 & 1.445 & 1.718 &  & 0.024 & 0.176 & 1.263 & 1.613 &  & 0.021 & 0.387 & 1.008 & 1.194 \\
& & j=10 & & 0.037 & 0.060 &  & 0.006 & 0.093 &  & 0.005 & 0.152 & & 0.048 & 0.093 & 1.854 & 1.642 &  & 0.019 & 0.150 & 1.130 & 1.157 &  & 0.009 & 0.291 & 0.792 & 0.825 \\
& & j=d & & 0.018 & 0.049 &  & 0.009 & 0.087 &  & 0.005 & 0.116 & & 0.028 & 0.090 & 0.529 & 0.476 &  & 0.019 & 0.129 & 0.532 & 0.609 &  & 0.008 & 0.210 & 0.367 & 0.436 \\  
\hline
\end{array}
$
\caption{\small RMSEs of $\hat{\theta}_j = (\eta_j, r_j)$ and $\hat{\theta}_j = (\eta_j, r_j, \beta_{1j}, \beta_{2j})$, the parameter vectors of the extremal-$t$ and extremal skew-$t$ models using the $j$-wise composite likelihood when $\nu =1$ and $d=20$. The case $j=d$ corresponds to full likelihood estimation using approximation Type II from Table~\ref{tab:approx}. Calculations are based on 500 replicate maximisations.}
\label{tab:RMSE_ET_SKT}
\end{table}
\end{landscape}
}

\afterpage{
\clearpage
\begin{landscape}
\begin{table}
\centering
\footnotesize
$
\begin{array}{ccc c cccccccc c cccccccccccccc}
\multicolumn{4}{c}{} & \multicolumn{8}{c}{\textrm{extremal-$t$}} & & \multicolumn{14}{c}{\textrm{extremal skew-$t$}} \\
\cline{5-12} \cline{14-27} 
\multicolumn{4}{c}{} & \multicolumn{2}{c}{r=1.5} & & \multicolumn{2}{c}{r=3.0} & & \multicolumn{2}{c}{r=4.5} & & \multicolumn{4}{c}{r=1.5} & & \multicolumn{4}{c}{r=3.0} & & \multicolumn{4}{c}{r=4.5} \\
\cline{5-6} \cline{8-9} \cline{11-12} \cline{14-17} \cline{19-22} \cline{24-27}
& & j &  & \hat{\eta}_j & \hat{r}_j & & \hat{\eta}_j & \hat{r}_j & & \hat{\eta}_j & \hat{r}_j & & \hat{\eta}_j & \hat{r}_j & \hat{\beta}_{1j} & \hat{\beta}_{2j} & & \hat{\eta}_j & \hat{r}_j & \hat{\beta}_{1j} & \hat{\beta}_{2j} & & \hat{\eta}_j & \hat{r}_j & \hat{\beta}_{1j} & \hat{\beta}_{2j} \\
\hline
 & \eta=1.00 & j=2 & & 0.003 & 0.008 &  & 0.015 & 0.004 &  & 0.006 & 0.099 & & 0.494 & 1.819 & 4.181 & 4.498 &  & 0.612 & 0.699 & 4.094 & 4.316 &  & 0.526 & 1.092 & 3.628 & 4.222  \\
& & j=3 & & 0.008 & 0.019 &  & 0.002 & 0.053 &  & 0.001 & 0.068 & & 0.240 & 0.286 & 0.334 & 1.205 &  & 0.189 & 0.200 & 0.094 & 1.158 &  & 0.149 & 0.326 & 0.243 & 1.309 \\
& & j=4 & & 0.010 & 0.000 &  & 0.009 & 0.036 &  & 0.009 & 0.114 & & 0.026 & 0.009 & 0.496 & 0.880 &  & 0.021 & 0.269 & 0.233 & 0.889 &  & 0.004 & 1.177 & 0.227 & 1.034 \\
& & j=5 & & 0.015 & 0.048 &  & 0.003 & 0.077 &  & 0.001 & 0.134 & & 0.118 & 0.071 & 0.062 & 1.295 &  & 0.073 & 0.139 & 0.152 & 1.492 &  & 0.048 & 0.028 & 0.045 & 1.393 \\
& & j=10 & & 0.001 & 0.011 &  & 0.000 & 0.013 &  & 0.001 & 0.061 & & 0.087 & 0.157 & 0.350 & 0.600 &  & 0.082 & 0.398 & 0.026 & 0.404 &  & 0.086 & 0.746 & 0.074 & 0.397 \\
& & j=d & & 0.008 & 0.021 &  & 0.007 & 0.050 &  & 0.004 & 0.072 & & 0.017 & 0.116 & 0.068 & 0.048 &  & 0.018 & 0.173 & 0.096 & 0.128 &  & 0.017 & 0.117 & 0.149 & 0.213 \\
\hline
& \eta=1.50 & j=2 & & 0.005 & 0.005 &  & 0.006 & 0.026 &  & 0.001 & 0.013 & & 0.375 & 0.910 & 4.304 & 4.659 &  & 0.274 & 1.736 & 3.795 & 4.733 &  & 0.245 & 2.302 & 3.857 & 4.508 \\
& & j=3 & & 0.001 & 0.014 &  & 0.002 & 0.022 &  & 0.001 & 0.015  & & 0.126 & 0.174 & 0.302 & 1.171 &  & 0.056 & 0.364 & 0.498 & 1.272 &  & 0.018 & 0.816 & 0.521 & 1.161 \\
& & j=4 & & 0.004 & 0.004 &  & 0.006 & 0.026 &  & 0.006 & 0.023 & & 0.000 & 0.036 & 0.169 & 0.776 &  & 0.023 & 0.306 & 0.040 & 1.001 &  & 0.040 & 0.970 & 0.279 & 1.044 \\
& & j=5 & & 0.015 & 0.033 &  & 0.004 & 0.041 &  & 0.004 & 0.033 & & 0.052 & 0.120 & 0.056 & 0.920 &  & 0.036 & 0.071 & 0.120 & 1.193 &  & 0.016 & 0.170 & 0.099 & 1.245 \\
& & j=10 & & 0.003 & 0.000 &  & 0.001 & 0.003 &  & 0.000 & 0.000 & & 0.027 & 0.055 & 0.351 & 0.468 &  & 0.027 & 0.138 & 0.009 & 0.336 &  & 0.024 & 0.396 & 0.141 & 0.400 \\
& & j=d & & 0.005 & 0.015 &  & 0.004 & 0.021 &  & 0.002 & 0.013 & & 0.010 & 0.051 & 0.072 & 0.059 &  & 0.011 & 0.045 & 0.055 & 0.046 &  & 0.004 & 0.123 & 0.068 & 0.048 \\ 
\hline
& \eta=1.95 & j=2 & & 0.011 & 0.001 &  & 0.008 & 0.008 &  & 0.004 & 0.033 & & 0.043 & 1.187 & 4.268 & 4.937 &  & 0.007 & 1.693 & 4.080 & 4.820 &  & 0.011 & 1.875 & 3.812 & 4.292 \\
& & j=3 & & 0.010 & 0.005 &  & 0.005 & 0.012 &  & 0.002 & 0.031 & & 0.010 & 0.329 & 0.777 & 1.532 &  & 0.021 & 0.298 & 0.598 & 1.754 &  & 0.030 & 0.658 & 0.405 & 1.858 \\
& & j=4 & & 0.010 & 0.003 &  & 0.002 & 0.001 &  & 0.001 & 0.013 & & 0.006 & 0.004 & 0.179 & 0.710 &  & 0.019 & 0.131 & 0.049 & 0.915 &  & 0.025 & 0.371 & 0.548 & 1.414 \\
& & j=5 & & 0.011 & 0.015 &  & 0.003 & 0.011 &  & 0.001 & 0.007 & & 0.002 & 0.053 & 0.055 & 0.834 &  & 0.002 & 0.076 & 0.077 & 0.860 &  & 0.006 & 0.253 & 0.037 & 0.466 \\
& & j=10 & & 0.004 & 0.003 &  & 0.000 & 0.000 &  & 0.000 & 0.015 & & 0.008 & 0.024 & 0.193 & 0.763 &  & 0.001 & 0.068 & 0.154 & 0.461 &  & 0.001 & 0.208 & 0.064 & 0.284 \\
& & j=d & & 0.000 & 0.006 &  & 0.001 & 0.004 &  & 0.001 & 0.002 & & 0.001 & 0.022 & 0.013 & 0.011 &  & 0.003 & 0.019 & 0.104 & 0.143 &  & 0.002 & 0.057 & 0.061 & 0.077 \\  
\hline
\end{array}
$
\caption{\small Absolute biases $|\bar{\hat{\theta}}_j - \theta_j|$ of $\hat{\theta}_j = (\eta_j, r_j)$ and $\hat{\theta}_j = (\eta_j, r_j, \beta_{1j}, \beta_{2j})$, the parameter vectors of the extremal-$t$ and extremal skew-$t$ models using the $j$-wise composite likelihood when $\nu =1$ and $d=20$. The case $j=d$ corresponds to full likelihood estimation using approximation Type II from Table~\ref{tab:approx}. Calculations are based on 500 replicate maximisations.}
\label{tab:BIAS_ET_SKT}
\end{table}
\end{landscape}
}

\clearpage
\section{Exact simulation of Extremal skew-\texorpdfstring{$t$}{Lg} Max Stable Process with Hitting Scenarios}
\label{app:algs}
 
Below we provide pseudo-code for exact simulation of extremal skew-$t$ max stable processes with unit Fr\'{e}chet marginal distributions using Algorithm 2 of \citet{dombry2016}, extended to include the hitting scenario in the output. This requires the simulation of an extended skew-$t$ distribution; here we use rejection sampling and the stochastic representation given in \citet{arellano2010}. The simpler extremal-$t$ max stable process only requires the simulation of a multivariate $t$-distribution and therefore does not use rejection sampling; this simpler algorithm is also given below.

When simulating $N$ independent replicates for $d$ sites with the \citet{dombry2016} algorithm, it is much more efficient to have the sites in the outer loop and the replicates in the inner loop, because derivations of quantities from the distribution of $(W(s) / W(s_0))_{s \in \cS}$ are then only performed once for each site (lines 3 to 7 in the skew-$t$ code), irrespective of the number of replicates required. In practice these quantities should be calculated on the log scale to avoid numerical issues.

In the algorithm below, the input $\Sigma_d$ is derived from the correlation function $\rho(h)$. The normalization in line 6 is needed for the simulation of an extended skew-$t$ distribution. Matrix multiplication is not needed here because $\omega$ is a diagonal matrix. The term $\mbox{Exp}(1)$ refers to a standard exponential distribution, $t_{\nu_d}$ is a univariate $t$-distribution with $\nu_d$ degrees of freedom, $N(0,1)$ is a standard univariate normal distribution, and $\chi^2_{\nu_d}$ is a chi-squared distribution with $\nu_d$ degrees of freedom. The function $\Psi(\cdot;\nu_d)$ is the distribution function of a univariate $t$-distribution, as used in equation \eqref{eq:mi_plus}. The code in line 16 simulates from a multivariate $t$-distribution with shape matrix $\Sigma^*_d$ and $\nu_d$ degrees of freedom. Lines 20 and 24 are identical by intent.

The index $j$ in the code corresponds to the $s_0$ site. We recommend the use of the eigendecomposition, which is more stable than the Cholesky decomposition. Moreover, $\Sigma^*_d$ is positive semi-definite as the $j$th row and columns are zero by construction. If the Cholesky decomposition were used then the code would need to handle the singular component explicitly. The eigendecomposition is slower, but it can be evaluated outside the loop over the observations (in line 7) and therefore only $d$ decompositions are required for any $N$.

The do-while loop in line 12 is the rejection sampling needed to simulate from  a multivariate extended skew-$t$ distribution. The \citet{dombry2016} algorithm also has a rejection step, with $B = 0$ in the code indicating rejection via exceeding an observation on an already simulated site (i.e.~on a site with index less than $j$). If the simulation is not rejected (line 22) then the outputs are set. A simulated process will always update the value on the $j$th site, because there is a singular component $X[j] \equiv 1$ and therefore the code would otherwise not enter the while loop at line 10. If the code enters the while loop (line 10), it breaks out of it when $\tilde{E}$ is small enough that the $j$th site simulation can never exceed the existing value. The vector $V$ counts the number of times the while loop executes for each replicate. This ultimately provides the hitting scenario $H$.  

\begin{algorithm}
       \caption{Extremal Skew-$t$ Process ($N$ Replicates)}
       \textbf{Inputs}: Correlation $\Sigma_d \in \real^{d\times d}$, Skew $\alpha \in \real^d$, DoF $\nu \in \real_{+}$ \;
       \textbf{Outputs}: Replicates $Z \in \real^{N\times d}$, Hitting Scenarios $H \in \real^{N\times d}$ \;
       \nextnr \textbf{initialize} outputs at $Z = -\infty$, $H = 0$ and \textbf{initialize} $V = 0 \in \real^{N}$ \;
       \nextnr \For{$j = 1$ \To{} $d$}{
       \nextnr Set $\bar{\Sigma}_d = \Sigma_d - \Sigma_{d;j} \Sigma_{j;d}$ and $\alpha^*_j$ = $\Sigma_{j;d}\alpha / (1 + \alpha^T \bar{\Sigma}_d \alpha)^{1/2}$ \;
       \nextnr Set  $\nu_d = \nu+1$ and $m_{j+} = 2^{\nu/2} \pi^{-1/2} \Gamma(\nu_d/2)
       \Psi(\alpha^*_j \nu_d^{1/2}; \nu_d)$ \;
       \nextnr Set $\mu_d = \Sigma_{d;j}$ and  $\tau_d = \nu_d^{1/2}\Sigma_{j;d} \alpha $ and $\hat{\Sigma}_d  = \bar{\Sigma}_d/\nu_d$   \;
       \nextnr         Set $\hat{\omega} = \diag(\hat{\Sigma}_d)^{1/2}$ and $\hat{\omega}_d = \diag(\Sigma_d)^{1/2}$ and $\hat{\alpha}_d = \sqrt{\nu_d} \hat{\omega}_d \alpha$  and $\hat{\bar{\Sigma}}_d = \hat{\omega}^{-1}\hat{\Sigma}_d \hat{\omega}^{-1}$ \;
       \nextnr         Calculate the eigen decomposition $\hat{\bar{\Sigma}}_d = L\Lambda^2L^T$ \;
       \nextnr         \For{$i = 1$ \To{} $N$}{
       \nextnr                 Simulate $E \sim \mbox{Exp}(1)$ and Set $\tilde{E} = (E/m_{j+})^{-1/\nu}$ \;
       \nextnr                 \While{$\tilde{E} > Z[i,j]$}{
       \nextnr                   Set $V[i] = V[i] + 1$\;
       \nextnr                   \Do{$t_{\nu_d} \geq \tau_d + \sum_{z=1}^d \hat{\alpha}_{d,z} Y_z$}{
       \nextnr               \For{$k = 1$ \To{} $d$}{
       \nextnr         Simulate $X[k] \sim N(0,1)$ \;
       }
       \nextnr               Set $Y = (Y[1],\dots, Y[d]) = \sqrt{\nu_d/\chi^2_{\nu_d}} L\Lambda X$  \;
                    }
       \nextnr                   Set $Y = \mu_d + \hat{\omega} Y$ and $B = 1$ \;
               \nextnr           \For{$k = 1$ \To{} $j-1$}{
               \nextnr                   \If{$Y[k] \tilde{E} > Z[i,j]$}{
               \nextnr                   Set $B = 0$ and \textbf{break} \;
                           }
                         }
       \nextnr               \If{$B=1$}{
       \nextnr                 \For{$k = j$ \To{} $d$}{
       \nextnr                 \If{$Y[k] \tilde{E} > Z[i,j]$}{
       \nextnr              Set $Z[i,j] = Y[k] \tilde{E}$ and $H[i,j] = V[i]$\;
                       }

                     }
               }
    \nextnr      Set $E = E + \mbox{Exp}(1)$ and $\tilde{E} = (E/m_{j+})^{-1/\nu}$ \;
       }
               }

       }
\nextnr Set $Z = Z^{\nu} m^{-1}$ for column vector $m = (m_{1+},\dots,m_{d+})$ \;
\end{algorithm}

\begin{algorithm}
       \caption{Extremal-t Process ($N$ Replicates)}
       \textbf{Inputs}: Correlation $\Sigma_d \in \real^{d\times d}$,  DoF $\nu \in \real_{+}$ \;
       \textbf{Outputs}: Replicates $Z \in \real^{N\times d}$, Hitting Scenarios $H \in \real^{N\times d}$ \;
       \nextnr \textbf{initialize} outputs at $Z = -\infty$, $H = 0$ and \textbf{initialize} $V = 0 \in \real^{N}$ \;
       \nextnr \For{$j = 1$ \To{} $d$}{
               \nextnr Set $\mu_d = \Sigma_{d;j}$ and $\nu_d = \nu+1$ and $\hat{\Sigma}_d  = (\Sigma_d - \Sigma_{d;j} \Sigma_{j;d})/\nu_d$ \;
               \nextnr         Calculate the eigendecomposition $\hat{\Sigma}_d = L\Lambda^2L^T$ \;
               \nextnr         \For{$i = 1$ \To{} $N$}{
                       \nextnr                 Simulate $E \sim \mbox{Exp}(1)$ and Set $\tilde{E} = E^{-1/\nu}$ \;
                       \nextnr                 \While{$\tilde{E} > Z[i,j]$}{
                               \nextnr                   Set $V[i] = V[i] + 1$\;
                                       \nextnr               \For{$k = 1$ \To{} $d$}{
                                               \nextnr         Simulate $X[k] \sim N(0,1)$ \;
                                       }
                                       \nextnr               Set $Y = (Y[1],\dots, Y[d]) = \mu_d + \sqrt{\nu_d/\chi^2_{\nu_d}} L\Lambda X$  \;
                               \nextnr                   Set $B = 1$ \;
                               \nextnr           \For{$k = 1$ \To{} $j-1$}{
                                       \nextnr                   \If{$Y[k] \tilde{E} > Z[i,j]$}{
                                               \nextnr                   Set $B = 0$ and \textbf{break} \;
                                       }
                               }
                               \nextnr               \If{$B=1$}{
                                       \nextnr                 \For{$k = j$ \To{} $d$}{
                                               \nextnr                 \If{$Y[k] \tilde{E} > Z[i,j]$}{
                                                       \nextnr              Set $Z[i,j] = Y[k] \tilde{E}$ and $H[i,j] = V[i]$\;
                                               }

                                       }
                               }
                               \nextnr      Set $E = E + \mbox{Exp}(1)$ and $\tilde{E} = E^{-1/\nu}$ \;
                       }
               }

       }
       \nextnr Set $Z = Z^{\nu}$ \;
\end{algorithm}

\end{document}